\documentclass[11pt]{article}
\usepackage{fullpage}
  
\usepackage{graphicx}  
\usepackage{subfigure} 
\usepackage{color}
\usepackage{amsmath, amsthm, amsfonts}
\usepackage{algpseudocode}
\usepackage{algorithm}
\usepackage{hyperref}
\usepackage{mathtools}
\usepackage{bm}
\usepackage{nicefrac}

\newtheorem{theorem}{Theorem}[section]
\newtheorem{lemma}[theorem]{Lemma}
\newtheorem{corollary}[theorem]{Corollary}

\theoremstyle{definition}

\newtheorem{problem}[theorem]{Problem}
\newtheorem{conjecture}[theorem]{Conjecture}

\usepackage[most]{tcolorbox}
\usepackage{empheq}
  
\newtcbox{\mymath}[1][]{%
    nobeforeafter, math upper, tcbox raise base,
    enhanced, colframe=blue!30!black,
    colback=blue!30, boxrule=1pt,
    #1}
  
\usepackage{color}
\newcommand{\field}[1]{\mathbb{#1}}
\newcommand{\hide}[1]{}

\newcommand{\squishlist}{
 \begin{list}{$\bullet$}
  {  \setlength{\itemsep}{0pt}
     \setlength{\parsep}{3pt}
     \setlength{\topsep}{3pt}
     \setlength{\partopsep}{0pt}
     \setlength{\leftmargin}{2em}
     \setlength{\labelwidth}{1.5em}
     \setlength{\labelsep}{0.5em}
} }
\newcommand{\squishlisttight}{
 \begin{list}{$\bullet$}
  { \setlength{\itemsep}{0pt}
    \setlength{\parsep}{0pt}
    \setlength{\topsep}{0pt}
    \setlength{\partopsep}{0pt}
    \setlength{\leftmargin}{2em}
    \setlength{\labelwidth}{1.5em}
    \setlength{\labelsep}{0.5em}
} }

\newcommand{\squishdesc}{
 \begin{list}{}
  {  \setlength{\itemsep}{0pt}
     \setlength{\parsep}{3pt}
     \setlength{\topsep}{3pt}
     \setlength{\partopsep}{0pt}
     \setlength{\leftmargin}{1em}
     \setlength{\labelwidth}{1.5em}
     \setlength{\labelsep}{0.5em}
} }

\newcommand{\squishend}{
  \end{list}
}

\newcommand{\squishlistt}{
 \begin{list}{---}
  {  \setlength{\itemsep}{0pt}
     \setlength{\parsep}{3pt}
     \setlength{\topsep}{3pt}
     \setlength{\partopsep}{0pt}
     \setlength{\leftmargin}{2em}
     \setlength{\labelwidth}{1.5em}
     \setlength{\labelsep}{0.5em}
} }

\def\defeq{\stackrel{\mathrm{def}}{=}}

\newcommand{\spara}[1]{\smallskip\noindent{\bf #1}}

\AtBeginDocument{%
  \providecommand\BibTeX{{%
    \normalfont B\kern-0.5em{\scshape i\kern-0.25em b}\kern-0.8em\TeX}}}

\tikzset{mynode/.style={inner sep=1pt,fill,outer sep=0,circle}}
\tikzset{font=\small}

\DeclarePairedDelimiterX{\gnorm}[2]{\lVert}{\rVert_{#2}}{#1}

\def\eps{\varepsilon}

\def\PP{\mathcal{P}}

\def\Bbf{\mathbf{B}}

\def\bbf{\mathbf{b}}
\def\fbf{\mathbf{f}}

\def\wbf{\mathbf{w}}\def\xbf{\mathbf{x}}
\def\ybf{\mathbf{y}}

\def\Rbb{\mathbb{R}}

\newcommand{\one}{\mathbf{1}}
\newcommand{\zero}{\mathbf{0}}

\begin{document}

\title{Flowless: Extracting Densest Subgraphs Without Flow Computations}
\author{
	Digvijay Boob\\
	Georgia Tech\\
	\texttt{digvijaybb40@gatech.edu}
	\and
	Yu Gao\\
	Georgia Tech\\
	\texttt{ygao380@gatech.edu}
	\and
	Richard Peng\\
	Georgia Tech\\
	\texttt{rpeng@cc.gatech.edu}
	\and
	Saurabh Sawlani\\
	Georgia Tech\\
	\texttt{sawlani@gatech.edu}
	\and
	Charalampos E. Tsourakakis\\
	Boston University\\
	\texttt{ctsourak@bu.edu}
	\and
	Di Wang\thanks{Work done when author was at Georgia Tech.}\\
	Google AI\\
	\texttt{wadi@google.com}
	\and
	Junxing Wang\\
	CMU\\
	\texttt{junxingw@cs.cmu.edu}
}

\maketitle
 
\begin{abstract}
We propose a simple and computationally efficient method for dense subgraph discovery in graph-structured data, which is a classic problem both in theory and in practice. It is well known that dense subgraphs can have strong correlation with structures of interest in real-world networks across various domains such as
social networks, communication systems, financial markets,
and biological systems \cite{gionis2015dense}. Consequently, this problem arises broadly in modern data science applications, and it is of great interest to design algorithms with practical appeal.

For the densest subgraph problem, which asks to find a subgraph with maximum average degree, Charikar's greedy algorithm \cite{asahiro2000greedily,charikar2000greedy} is guaranteed to find a $2$-approximate optimal solution. Charikar's algorithm is very simple, and can typically find result of quality much better than the provable factor $2$-approximation, which makes it very popular in practice. However, it is also known to give suboptimal output in many real-world examples. On the other hand, finding the exact optimal solution requires the computation of maximum flow \cite{goldberg1984finding,gallo1989fast,khuller2009finding}. Despite the existence of highly optimized maximum flow solvers, such computation still incurs prohibitive computational costs for the massive graphs arising in modern data science applications.

We devise a simple iterative algorithm which naturally generalizes the greedy algorithm of Charikar. Moreover, although our algorithm is fully combinatorial, it draws insights from the iterative approaches from convex optimization, and also exploits the dual interpretation of the densest subgraph problem. We have empirical evidence that our algorithm is much more robust against the structural heterogeneities in real-world datasets, and converges to the optimal subgraph density even when the simple greedy algorithm fails. On the other hand, in instances where Charikar's algorithm performs well, our algorithm is able to quickly verify its optimality. Furthermore, we demonstrate that our method is significantly faster than the maximum flow based exact optimal algorithm. We conduct experiments on real-world datasets from broad domains, and our algorithm achieves $\sim$145$\times$ speedup on average to find subgraphs whose density is at least 90\% of the optimal value.
\end{abstract}

\section{Introduction}  
\label{sec:intro} 
Finding dense components in graphs is a major topic in graph mining with diverse applications including DNA motif detection, unsupervised detection of interesting stories from micro-blogging streams in real time, indexing graphs 
for  efficient distance query computation, and anomaly detection in financial networks, and social networks \cite{gionis2015dense}. The {\em densest subgraph problem} (DSP) is one of the major formulations for dense subgraph discovery, where, given an undirected weighted graph $G(V,E,w)$  we want to find a set of nodes $S\subseteq V$ that maximizes the \emph{degree density} $\nicefrac{w(S)}{|S|}$, where $w(S)$ is the sum of the weights of the edges in the graph induced by $S$.  When the weights are non-negative, the problem is solvable in polynomial time using maximum flows \cite{goldberg1984finding}. Since maximum flow computations are expensive despite the theoretical progress achieved over the recent years, Charikar's greedy peeling algorithm is frequently used in practice \cite{charikar2000greedy}. This algorithm iteratively peels the lowest degree node from the graph, thus producing a sequence of  subsets of nodes, of which it outputs the densest one. This simple, linear time and linear space algorithm provides a $\nicefrac{1}{2}$-approximation for the DSP. However, when the edge weights are allowed to be negative, the DSP becomes NP-hard \cite{tsourakakis2019novel}.

Our work was originally motivated by a natural question: 
How can we quickly assess whether the output of Charikar's algorithm on a given graph instance is closer to optimality or to the worst case $\frac{1}{2}$-approximation guarantee?  However, we ended up answering the following intriguing question that we state as the next problem:

\begin{tcolorbox}
\begin{problem}
\label{prob1} 
Can we design an algorithm that  performs (i)  as well as Charikar's greedy algorithm in terms of efficiency, and (ii) as well as the maximum flow-based exact algorithm in terms of output quality?
\end{problem}
\end{tcolorbox}

\spara{Contributions.} The contributions of this paper are summarized as follows:

$\bullet$ We design a novel algorithm {\sc Greedy++} for the densest subgraph problem, a major dense subgraph discovery primitive that ``lies at the heart of large-scale data mining'' \cite{bahmani2012densest}. {\sc Greedy++} combines the best of two different worlds, the accuracy of the exact maximum flow based algorithm \cite{goldberg1984finding,gallo1989fast}, and the efficiency of Charikar's greedy peeling algorithm \cite{charikar2000greedy}.  

$\bullet$ It is worth outlining that Charikar's greedy algorithm typically performs better on real-world graphs than the worse case $\nicefrac{1}{2}$-approximation; on a variety of datasets we have tried, the worst case approximation was 0.8.  Nonetheless, the only way to verify how close the output is to optimality relies on  computing the exact solution using maximum flow. Our proposed method {\sc Greedy++} can be used to assess the accuracy of Charikar's algorithm in practice. Specifically, we find empirically that for all graph instances where {\sc Greedy++} after a couple of iterations does not significantly improve the output density, the output of Charikar's algorithm is near-optimal.

$\bullet$ We implement our proposed algorithm in C++ and apply it on a variety of real-world datasets. We verify the practical value of {\sc Greedy++}.  Our empirical results indicate that {\sc Greedy++}  is a valuable addition to the toolbox of dense subgraph discovery; on real-world graphs, {\sc Greedy++} is both fast in practice, and converges to a solution with an arbitrarily small approximation factor.

\spara{Notation.} Let $G(V,E)$ be a undirected graph, where $|V|=n, |E|=m$. For a given subset of nodes $S \subset V$, $e[S]$ denotes the number of edges induced by $S$. When the graph is weighted, i.e., there exists a weight function $w:E \mapsto \field{R}^+$, and $w(S)$ denotes the sum of the weights of the edges induced by $S$.  We use $N(u)$ to define the set of neighbors of $u$, and $\deg(u)=|N(u)|$. We use $\deg_S(u)$ to denote $u$'s degree in $S$, i.e., the number of neighbors of $u$ within  the set of nodes $S$. We use $\deg_{\max}$ to denote the maximum degree in $G$.
Finally, the \emph{degree density} $\rho(S)$ of a vertex set $S \subseteq V$ is defined as  $\frac{e[S]}{|S|}$, or $\frac{w(S)}{|S|}$ when the graph is weighted.

\section{Related Work}  
\label{sec:related} 
\spara{Dense subgraph discovery.} Detecting dense components is a major problem in graph mining. It is not surprising that many different notions of a dense subgraph are used in practice. The prototypical dense subgraph is a clique. However, the maximum clique problem is not only NP-hard, but also strongly inapproximable, see \cite{hastad}.  The notion of optimal quasi-cliques has been developed to detect subgraphs that are not necessarily fully interconnected but very dense \cite{tsourakakis2013denser}. However, finding optimal quasi-cliques is also NP-hard \cite{kawase2018densest,tsourakakis2015streaming}.   Another popular and scalable approach to finding dense components is based on $k$-cores \cite{esfandiari2018parallel}. Recently, $k$-cores have also been used to detect anomalies in large-scale networks \cite{giatsidis2014corecluster,shin2016corescope}.  

The interested reader may refer to the recent survey by Gionis and Tsourakakis on the more broad topic of dense subgraph discovery \cite{gionis2015dense}.  In the following, we only provide a brief overview of work related to the densest subgraph problem. 
 
\spara{Densest subgraph problem (DSP).} The goal of the {\em densest subgraph problem} (DSP) is to find the set of nodes $S$ which maximizes the degree density $\rho(S)$.
The densest subgraph can be identified in polynomial time by solving a maximum flow problem \cite{gallo1989fast,khuller2009finding,goldberg1984finding}.  
Charikar \cite{charikar2000greedy} proved that the greedy algorithm proposed by Asashiro et al.\footnote{Despite the fact that the greedy algorithm was originally proposed in \cite{asahiro2000greedily}, it is widely known as Charikar's greedy algorithm.}
\cite{asahiro2000greedily} produces a $\nicefrac{1}{2}$-approximation of the densest subgraph in linear time.
To obtain fast algorithms with better approximation factors,
McGregor et. al. \cite{mcgregor2015densest}, and Mitzenmacher et. al. \cite{mitzenmacher2015scalable}
uniformly sparsified the input graph, and computed the densest subgraph in the resulting sparse graph.
The first near-linear time algorithm for the DSP, given by Bahmani et. al. \cite{bahmani2014efficient},
relies on approximately solving the LP dual of the DSP.
It is worth mentioning that 
Kannan and Vinay \cite{kannan1999analyzing} gave a spectral $O(\log{n})$ approximation algorithm for a related notion of density.

\spara{Charikar's greedy peeling algorithm.} Since our algorithm {\sc Greedy++} is an improvement over Charikar's greedy algorithm, we discuss the latter algorithm in greater detail.  The algorithm removes in each iteration, the node with the smallest degree. This process creates a nested sequence of sets of nodes $V=S_n \supset S_{n-1} \supset S_{n-2} \supset \ldots \supset S_1 \supset \emptyset$. The algorithm outputs the graph $G[S_j]$ that maximizes the degree density among  $j=1,\ldots,n$. The pseudocode is shown in Algorithm~\ref{alg:greedy}.

\begin{algorithm}
	\caption{ {\sc Greedy}}
	\label{alg:greedy}
	\begin{flushleft}
		\textbf{Input}: Undirected graph $G$ 
		
		\textbf{Output}:  A dense subgraph of $G$: $G_{\textsf{densest}}$.
	\end{flushleft}
	\begin{algorithmic}[1]
		\State $G_{\textsf{densest}} \leftarrow G$
		\State $H \leftarrow G$;
		\While{$H \neq \emptyset$}
		\State Find the vertex $u \in H$ with minimum $\textsf{deg}_H(u)$;
		\State Remove $u$ and all its adjacent edges $uv$ from $H$;
		\If{$\rho(H) > \rho(G_{\textsf{densest}})$}
		\State $G_{\textsf{densest}} \leftarrow H$
		\EndIf
		\EndWhile
		\State Return $G_{\textsf{densest}}$.
	\end{algorithmic}
\end{algorithm} 

\spara{Fast numerical approximation algorithms for DSP.} Bahmani et. al. \cite{bahmani2014efficient} approached the DSP via its dual problem, which in turn they reduced to $O(\log n)$
instances of solving a positive linear program.
To solve these LPs, they employed the multiplicative
weights update framework \cite{arora2012multiplicative, plotkin1995fast}
to achieve an $\eps$-approximation in ${O}(\log n / \eps^2)$ iterations,
where each iteration requires $O(m)$ work.

\spara{Notable extensions of the DSP.} The DSP has been studied in weighted graphs, as well as directed graphs. When the edge weights are non-negative, both the maximum flow algorithm and Charikar's greedy algorithm maintain their theoretical guarantees. In the presence of negative weights, the DSP in general becomes NP-hard \cite{tsourakakis2019novel}.   For directed graphs Charikar \cite{charikar2000greedy} provided a linear programming approach which requires the computation of $n^2$ linear programs
and a $\nicefrac{1}{2}$-approximation algorithm which runs in $O(n^3+n^2m)$ time. Khuller and Saha have provided more efficient implementations of the exact and approximation algorithms for the undirected and directed versions of the DSP  \cite{khuller2009finding}.   Furthermore, Tsourakakis et al. recently extended the DSP to the $k$-clique, and the $(p,q)$-biclique densest subgraph problems \cite{tsourakakis2015k,mitzenmacher2015scalable}. These extensions can be used for finding large near-cliques in general graphs and bipartite graphs  respectively.  The DSP has also been studied in the dynamic setting \cite{bhattacharya2015space,epasto2015efficient,sawlani2019near}, 
the streaming setting \cite{bahmani2012densest,bhattacharya2015space,mcgregor2015densest,esfandiari2015applications}, and in the MapReduce computational model \cite{bahmani2012densest}.  Bahmani, Goel, and Munagala use the  multiplicative weights update framework \cite{arora2012multiplicative,plotkin1995fast} to design an  improved MapReduce algorithm \cite{bahmani2014efficient}. We discuss this method in greater detail in Section~\ref{sec:prop}.
Tatti and Gionis \cite{tatti2015density} introduced a novel graph decomposition known as \emph{locally-dense}, that imposes certain insightful constraints on the k-core decomposition.
Further, efficient algorithms to find locally-dense subgraphs were developed by Danisch et al. \cite{danisch2017large}.

We notice that in the DSP there are no restrictions on the size of the output. When restrictions on the size of $S$ are imposed the problem becomes NP-hard.  The densest-$k$-subgraph problem asks for find the subgraph $S$ with maximum degree density among all possible sets $S$ such that $|S|=k$. The state-of-the art algorithm is due to Bhaskara et al. \cite{bhaskara2010detecting}, and provides a $O(n^{1/4+\epsilon})$ approximation in $O(n^{1/\epsilon})$ time. A long standing question is closing the gap between this upper bound and the lower bound. Other versions where $|S|\geq k, |S|\leq k$ have also been considered in the literature see \cite{andersen2009finding}.

\section{Proposed Method} 
\label{sec:prop} 
\subsection{The {\sc Greedy++} algorithm}
\label{subsec:greedyplus} 

As we discussed earlier, Charikar's peeling algorithm greedily removes the node of smallest degree from the graph, and returns the densest subgraph among the sequence of $n$ subgraphs created by this procedure. While ties may exist, and are broken arbitrarily, for the moment it is useful to think as if  Charikar's greedy algorithm produces a single permutation of the nodes, that naturally defines a nested sequence of subgraphs. 

\spara{Algorithm description.} Our proposed algorithm {\sc Greedy++} iteratively runs Charikar's peeling algorithm, while keeping some information about the past runs. This information is crucial, as it results in different permutations, that naturally yield higher quality outputs. The pseudocode for \textsc{Greedy++} is shown in Algorithm~\ref{alg:main}. It takes as input the graph $G$, and a parameter $T$ of the number of passes to be performed, and runs an iterative, weighted peeling procedure. In each round the load of each node is a function of its induced degree and the load from the previous rounds. It is worth outlining that the algorithm is easy to implement, as it is essentially $T$ instances of Charikar's algorithm. What is less obvious perhaps, is why this algorithm makes sense, and works well.  We answer this question in detail in Section~\ref{subsec:mwu}.

\begin{algorithm}
	\caption{ {\sc Greedy++}}
	\label{alg:main}
	\begin{flushleft}
		\textbf{Input}: Undirected graph $G$, iteration count $T$
		
		\textbf{Output}: An approximately densest subgraph of $G$: $G_{\textsf{densest}}$.
	\end{flushleft}\begin{algorithmic}[1]
		\State $G_{\textsf{densest}} \leftarrow G$
		\State Initialize the vertex load vector $\ell^{(0)} \leftarrow 0 \in \mathbb{Z}^n$;
		\For{$i : 1 \rightarrow T$}
		\State $H \leftarrow G$;
		\While{$H \neq \emptyset$}
		\State Find the vertex $u \in H$ with minimum $\ell^{(i-1)}_u + \textsf{deg}_H(u)$;
		\State $\ell^{(i)}_u \leftarrow \ell^{(i-1)}_u + \textsf{deg}_H(u)$;
		\State Remove $u$ and all its adjacent edges $uv$ from $H$;
		\If{$\rho(H) > \rho(G_{\textsf{densest}})$}
		\State $G_{\textsf{densest}} \leftarrow H$
		\EndIf
		\EndWhile
		\EndFor
		\State Return $G_{\textsf{densest}}$.
	\end{algorithmic}
\end{algorithm}

\spara{Example.} We provide a graph instance that clearly illustrates why {\sc Greedy++} is a significant improvement over the classical greedy algorithm. We discuss the first two rounds of {\sc Greedy++}. Consider the following graph $G=B \bigcup \left(\cup_{i=1}^k H_i\right)$ where $B = K_{d,D}$ and $H_i=K_{d+2}$.
Namely  $G$ is a disjoint union of a complete $d \times D$ bipartite graph $B$, and of $k$ $(d+2)$-cliques $H_1,\ldots,H_k$. Consider the case where  $d \ll D, k \rightarrow +\infty$. $G$ is pictured in Figure~\ref{fig:example}(a).
The density of $G$ is  
\[
\dfrac{2dD + (d+1)(d+2)k}{2d + 2D + 2k(d+2)} \rightarrow \frac{d+1}{2}. 
\] 
Notice that this is precisely the density of any $(d+2)$-clique. However, the density of $B$ is $\frac{dD}{d+D} \approx d$, which is in fact the optimal solution.
Charikar's algorithm outputs $G$ itself, since it starts eliminating nodes of degree $d$ from $B$,
and by doing this, it never sees a subgraph with higher density. This example illustrates that the $\frac{1}{2}$ approximation is tight.  Consider now a run of {\sc Greedy++}. 

In its first iteration, it simply emulates Charikar's algorithm. The $D-d$ vertices of $B$ which were eliminated first - each have load $d$. At this stage, our input is the disjoint union of $k$ cliques and  a $d\times d$ bipartite graph. Of the remaining $2d$ vertices in $B$, one vertex is charged with load $d$,
two vertices each with loads $(d-1),(d-2),\ldots,1$,  and one vertex with load $0$.
On the other hand, vertices in $H_i$ are charged with loads $d+1, d, \ldots, 0$.
Figure~\ref{fig:example}(b) shows the cumulative degrees of vertices in $G$ after one iteration
of \textsc{Greedy++}.

 Without any loss of generality let us assume the vertex from $B$ that got charged $0$ originally had degree $d$.  This vertex in the second iteration will  get deleted first, and the vertex whose sum of load and degree is $d+1$ will get deleted second.   But after these two, all the cliques get peeled away by the algorithm.
This leaves us with a $d \times D-2$ bipartite graph as the output after the second iteration, whose density is almost optimal.

\begin{figure}[!ht]
	\centering
	
	\begin{tabular}{@{}c@{}}
		\begin{tikzpicture}[scale=0.37]
		\node at (2,9) {$\bm{B = K_{d,D}}$};
		
		\coordinate (1) at (0,-3);
		\coordinate (2) at (0,-1.5);
		\coordinate (3) at (0,0);
		\node at (0,1.5) {$\vdots$};
		\coordinate (4) at (0,3);
		
		\coordinate (5) at (4,-3);
		\coordinate (6) at (4,-1.5);
		\coordinate (7) at (4,0);
		\node at (4,1.5) {$\vdots$};
		\coordinate (8) at (4,3);
		\coordinate (9) at (4,4.5);
		\node at (4,6) {$\vdots$};
		\coordinate (10) at (4,7.5);
		
		\node[mynode, label=left:{$(D)~a_1$}] at (1) {};
		\node[mynode, label=left:{$(D)~a_2$}] at (2) {};
		\node[mynode, label=left:{$(D)~a_3$}] at (3) {};
		\node[mynode, label=left:{$(D)~a_d$}] at (4) {};
		\node[mynode, label=right:{$b_1~(d)$}] at (5) {};
		\node[mynode, label=right:{$b_2~(d)$}] at (6) {};
		\node[mynode, label=right:{$b_3~(d)$}] at (7) {};
		\node[mynode, label=right:{$b_d~(d)$}] at (8) {};
		\node[mynode, label=right:{$b_{d+1}~(d)$}] at (9) {};
		\node[mynode, label=right:{$b_D~(d)$}] at (10) {};
		
		\foreach \x in {1,2,...,4}{
			\foreach \y in {5,6,...,10}{
				\draw[line width=0.2pt] (\x) -- (\y);
			}
		}
		
		\node at (14,9) {$\bm{H_i = K_{d+2}}$};
		
		\coordinate (11) at (14,0);
		\coordinate (12) at (16,1);
		\coordinate (13) at (16,3);
		\coordinate (14) at (14,4);
		\coordinate (15) at (12,3);
		\node at (12,2.25) {$\vdots$};
		\coordinate (16) at (12,1);
		
		\node[mynode, label=below:{\begin{tabular}{c} $c_1$ \\ $(d+1)$ \end{tabular}}] at (11) {};
		\node[mynode, label=right:{$c_2$}] at (12) {};
		\node at (17,0) {$(d+1)$};
		\node[mynode, label=right:{$c_3$}] at (13) {};
		\node at (17,4) {$(d+1)$};
		\node[mynode, label=above:{\begin{tabular}{c} $(d+1)$ \\ $c_4$ \end{tabular}}] at (14) {};
		\node[mynode, label=left:{$c_5$}] at (15) {};
		\node at (11,4) {$(d+1)$};
		\node[mynode, label=left:{$c_{d+2}$}] at (16) {};
		\node at (11,0) {$(d+1)$};
		
		\foreach \y in {12,...,16}{
			\draw[line width=0.2pt] (11) -- (\y);
		}
		\foreach \y in {13,...,16}{
			\draw[line width=0.2pt] (12) -- (\y);
		}
		\foreach \y in {14,15,16}{
			\draw[line width=0.2pt] (13) -- (\y);
		}
		\foreach \y in {15,16}{
			\draw[line width=0.2pt] (14) -- (\y);
		}
		
		\end{tikzpicture} \\[\abovecaptionskip]
		
		\small (a) Initial degrees of $G$
	\end{tabular}
	
	\vspace{\floatsep}
	
	\begin{tabular}{@{}c@{}}
		\begin{tikzpicture}[scale=0.37]
		\node at (2,9) {$\bm{B = K_{d,D}}$};
		
		\coordinate (1) at (0,-3);
		\coordinate (2) at (0,-1.5);
		\coordinate (3) at (0,0);
		\node at (0,1.5) {$\vdots$};
		\coordinate (4) at (0,3);
		
		\coordinate (5) at (4,-3);
		\coordinate (6) at (4,-1.5);
		\coordinate (7) at (4,0);
		\node at (4,1.5) {$\vdots$};
		\coordinate (8) at (4,3);
		\coordinate (9) at (4,4.5);
		\node at (4,6) {$\vdots$};
		\coordinate (10) at (4,7.5);
		
		\node[mynode, label=left:{$(D+1)~a_1$}] at (1) {};
		\node[mynode, label=left:{$(D+2)~a_2$}] at (2) {};
		\node[mynode, label=left:{$(D+3)~a_3$}] at (3) {};
		\node[mynode, label=left:{$(D+d)~a_d$}] at (4) {};
		\node[mynode, label=right:{$b_1~(d)$}] at (5) {};
		\node[mynode, label=right:{$b_2~(d+1)$}] at (6) {};
		\node[mynode, label=right:{$b_3~(d+2)$}] at (7) {};
		\node[mynode, label=right:{$b_d~(2d-1)$}] at (8) {};
		\node[mynode, label=right:{$b_{d+1}~(2d)$}] at (9) {};
		\node[mynode, label=right:{$b_D~(2d)$}] at (10) {};
		
		\foreach \x in {1,2,...,4}{
			\foreach \y in {5,6,...,10}{
				\draw[line width=0.2pt] (\x) -- (\y);
			}
		}
		
		\node at (14,9) {$\bm{H_i = K_{d+2}}$};
		
		\coordinate (11) at (14,0);
		\coordinate (12) at (16,1);
		\coordinate (13) at (16,3);
		\coordinate (14) at (14,4);
		\coordinate (15) at (12,3);
		\node at (12,2.25) {$\vdots$};
		\coordinate (16) at (12,1);
		
		\node[mynode, label=below:{\begin{tabular}{c} $c_1$ \\ $(2d+2)$ \end{tabular}}] at (11) {};
		\node[mynode, label=right:{$c_2$}] at (12) {};
		\node at (17,0) {$(2d+1)$};
		\node[mynode, label=right:{$c_3$}] at (13) {};
		\node at (17,4) {$(2d)$};
		\node[mynode, label=above:{\begin{tabular}{c} $(2d-1)$ \\ $c_4$ \end{tabular}}] at (14) {};
		\node[mynode, label=left:{$c_5$}] at (15) {};
		\node at (11,4) {$(2d-2)$};
		\node[mynode, label=left:{$c_{d+2}$}] at (16) {};
		\node at (11,0) {$(d+1)$};
		
		\foreach \y in {12,...,16}{
			\draw[line width=0.2pt] (11) -- (\y);
		}
		\foreach \y in {13,...,16}{
			\draw[line width=0.2pt] (12) -- (\y);
		}
		\foreach \y in {14,15,16}{
			\draw[line width=0.2pt] (13) -- (\y);
		}
		\foreach \y in {15,16}{
			\draw[line width=0.2pt] (14) -- (\y);
		}
		
		\end{tikzpicture} \\[\abovecaptionskip]
		\small (b) Cumulative degrees (degree + load) of $G$ after one iteration
	\end{tabular}
	
	\caption{\label{fig:example}Illustration of two iterations of \textsc{Greedy++} on $G$.
	The output after one iteration is $\bm{G}$ itself (density $\bm{\approx (d+1)/2}$),
	whereas the output after the
	second iteration is $\bm{B \setminus \{b_1, b_2\}}$ (density $\bm{\approx d}$).}
	\vspace{-4mm}
\end{figure}

\spara{Theoretical guarantees.} Before we prove our theoretical properties for our proposed algorithm  {\sc Greedy++}, it is worth outlining that experiments indicate that the performance of  {\sc Greedy++} is significantly better than the worst-case analysis we perform.  Furthermore, we conjecture that our guarantees are not tight from a theoretical perspective; an interesting open question is to extend our analysis in Section~\ref{subsec:mwu} for {\sc Greedy++} to prove that it provides asymptotically an optimal solution for the DSP. We conjecture that our algorithm is a $(1+\frac{1}{\sqrt{T}})$-approximation algorithm for the DSP. Our fist lemma states that {\sc Greedy++} is a $2$-approximation algorithm for the DSP. 

\begin{lemma}
\label{lem1} 
Let $G_{\textsf{densest}}$ the output of {\sc Greedy++}. Then, $\rho(G_{\textsf{densest}}) \geq \rho_G^*/2$, where $\rho_G^*$
denotes the optimum value of the problem.
\end{lemma}

\begin{proof}
Notice that the first iteration is identical to Charikar's $2$-approximation algorithm, and $G_{\textsf{densest}}$ is at least as dense as the output of the first iteration. \qedhere
\end{proof}

The next lemma provides bounds the quality of the dual solution, i.e., at each iteration the average load (average over the algorithm's iterations) assigned to any vertex is at most $2\rho_G^*$.

\begin{lemma}
The following invariant holds for {\sc Greedy++}: for any vertex $v$ and iteration $i$, $\ell_v^{(i)} \leq 2i \cdot \rho^*_G.$
\end{lemma}

\begin{proof}
	First, let $i=1$.
	The proof for this base case goes through identically as in
	\cite{charikar2000greedy}.
	\begin{align*}
	\ell_v^{(1)} = \deg_{G^{(1)}_v}(v) & \leq \dfrac{1}{|V_v^{(i)}|} \sum_{u \in V_v^{(i)}} \deg_{G^{(i)}_v}(u) \\
	& =\dfrac{2|E^{(i)}_v|}{|V_v^{(i)}|}\\
	& = 2 \cdot \rho_{G_v} \leq 2 \cdot \rho^*_G.
	\end{align*}
	
	Now, assume that the statement is true for some iteration index $i-1$.
	Consider the point at which vertex $v$ is chosen in iteration $i$.
	Denote the graph at that instant to be $G^{(i)}_v = \left\langle V^{(i)}_v, E^{(i)}_v \right\rangle$.
	For any vertex $u$ at that point, the cumulative degree is
	$\ell_u^{(i-1)} + \deg_{G^{(i)}_v}(u)$.
	Since $v$ has the minimum cumulative degree at that point,
	\begin{align*}
	\ell_v^{(i)} = \ell_v^{(i-1)} + \deg_{G^{(i)}_v}(v) & \leq \dfrac{1}{|V_v^{(i)}|} \sum_{u \in V_v^{(i)}} \left(\ell_u^{(i-1)} + \deg_{G^{(i)}_v}(u)\right) \\
	& \leq 2(i-1) \rho_G^* + \dfrac{1}{|V_v^{(i)}|} \sum_{u \in V_v^{(i)}} \deg_{G^{(i)}_v}(u)\\
	& \leq 2i \cdot \rho^*_G. \qedhere
	\end{align*}
\end{proof}

\spara{Running time.} Finally, we bound the runtime of the algorithm as follows. The next lemma states that our algorithm can be implemented to run in $O((n+ m) \cdot \min(\log n, T))$.

\begin{lemma}
	Each iteration of the above algorithm runs in time $O((n+m) \cdot \min(\log n, T))$.
\end{lemma}

\begin{proof}
	The deletion operation, along with assigning edges to a vertex
	and updating degrees takes $O(m)$ time since every edge is assigned once.	Finding the minimum degree vertex can be implemented in two ways:
	\begin{enumerate}
		\item Since degrees in our algorithm can go from $0$ to $2Tm$,
		we can create lists for each separate integer degree value.
		Now we need to scan each list from $\deg=1$ to $\deg=2Tm$.
		However, after deleting a vertex of degree $d$, we only need
		to scan from $d-1$ onwards. So the total time taken is $O(2Tm+n) = O(mT)$.
		\item We can maintain a priority queue, which needs a total of $O(m)$ update operations,
		each taking $O(\log n)$ time. \qedhere
	\end{enumerate}
\end{proof}

Note that in the case of weighted graphs, we cannot maintain lists for each possible degree,
and hence, it is necessary to use a priority queue.

\subsection{Why does  {\sc Greedy++} work well?} 
\label{subsec:mwu} 

Explaining the intuition behind  {\sc Greedy++}  requires an understanding of the load balancing interpretation of Charikar's LP for the DSP \cite{charikar2000greedy}, and the multiplicative weights update (MWU) framework  by Plotkin, Shmoys and Tardos \cite{plotkin1995fast} used for packing/covering LPs. In the context of the DSP, the MWU  framework was first used by Bahmani, Goel, and Munagala \cite{bahmani2014efficient}. We include a self-contained exposition of the required concepts from  \cite{bahmani2014efficient,charikar2000greedy} in this section, that has a natural flow and concludes with our algorithmic contributions. Intuitively, the additional passes that {\sc Greedy++} performs, improve the load balancing.   

\spara{Charikar's LP and the load balancing interpretation.} The following is a well-known LP formulation of the densest subgraph problem, introduced in \cite{charikar2000greedy},
which we denote by $\textsc{Primal}(G)$. The optimal objective value is known to be $\rho_G^*$.

\begin{equation*}
\begin{array}{ll@{}ll}
\text{maximize}  & \displaystyle\sum\limits_{e \in E} & y_e &\\
\text{subject to} & & y_e \leq x_u, \qquad & \forall e=uv \in E\\
&                  & y_e \leq x_v,        & \forall e=uv \in E\\
& \displaystyle\sum\limits_{v \in V} & x_v \leq 1,\\
&                  & y_e \geq 0,        & \forall e \in E\\
&                  & x_v \geq 0,        & \forall v \in V
\end{array}
\end{equation*}

We then construct the dual LP for the above problem.
Let $f_{e}(u)$ be the dual variable associated with the first $2m$ constraints of the form $y_e \leq x_u$,
and let $D$ be associated with the last constraint.
We get the following LP, which we denote by $\textsc{Dual}(G)$, and whose optimum is also $\rho_G^*$.

\begin{equation*}
\begin{array}{lr@{}ll}
\text{minimize}  & & D &\\
\text{subject to} & f_e(u) + & f_e(v) \geq 1, \qquad & \forall e=uv \in E\\
&\ell_v \defeq \displaystyle\sum\limits_{e \ni v} & f_e(v) \leq D, & \forall v \in V\\
&                  & f_e(u) \geq 0,        & \forall e=uv \in E\\
&                  & f_e(v) \geq 0,        & \forall e=uv \in E
\end{array}
\end{equation*}

This LP can be visualized as follows.  Each edge $e=uv$ has a load of $1$, which it wants to send to its end points: $f_e(u)$ and $f_e(v)$ such that the total load of any vertex $v$, $\ell_v$, is at most $D$.  The objective is to find the minimum $D$ for which such a load assignment is feasible.

For a fixed $D$, the above dual problem can be framed as a flow problem on a bipartite graph as follows:
Let the left side $L$ represent $V$ and the right side $R$ represent $E$. Add a super-source $s$ and edges from $s$ to all vertices in $L$ with capacity $D$. Add edges from $v \in V$ to $e \in E$ if $e$ is incident on $v$ in $G$. All vertices in $R$ have demands of $1$ unit.
Although Goldberg's initial reduction \cite{goldberg1984finding} involved a different flow network, this graph can also be used to use maximum flow and use that to find the exact optimum to our problem. From strong duality, we know that the optimal objective values
of both linear programs are equal, i.e., exactly $\rho_G^*$.
Let $\rho_G$ be the objective of any feasible solution to $\textsc{Primal}(G)$.
Similarly, let $\hat \rho_G$ be the objective of any feasible solution to $\textsc{Dual}(G)$. Then, by optimality of $\rho_G^*$ and weak duality, we obtain the optimality result $\rho_G \leq \rho_G^* \leq \hat\rho_G$.

\spara{Bahmani et al. \cite{bahmani2014efficient}}   use the following covering LP formulation: decide the feasibility of constraints $f_e(u) + f_e(v) \ge 1$ for each edge $e = uv \in E$ subject to the polyhedral constraints:
\begin{align*}
\sum_{e \ni v}f_e(v) &\le D, \qquad &\forall v \in V\\
f_e(u)&\ge 0, &\forall e = uv \in E\\
f_e(v)&\ge 0, &\forall e = uv \in E
\end{align*}

The width of this linear program is the maximum value of $f_e(u)+f_e(v)$ provided that $f_e(u), f_e(v)$ satisfy the constraints of the program.  Bahmani et al. in order to provably bound the \emph{width}  of the above LP,
they introduce another set of simple constraints as follows:
\begin{align*}
\sum_{e \ni v}f_e(v) &\le D, \qquad &\forall v \in V\\
q \ge f_e(u)&\ge 0, &\forall e = uv \in E\\
q \ge f_e(v)&\ge 0, &\forall e = uv \in E
\end{align*}
\noindent where $q \ge 1$ is a small constant.
So, for a particular value of $D$, they verify the approximate feasibility of the covering problem using the MWU framework.
However, this necessitates running a binary search over all possible
values of $D$ and finding the lowest value of $D$ for which the LP is feasible.
Since the precision for $D$ can be as low as $\epsilon$, this
binary search is inefficient in practice.
Furthermore, due to the added $\ell_{\infty}$ constraint to bound the width,
extracting the primal solution (i.e. an approximately densest subgraph)
from the dual is no longer straightforward, and the additional
rounding step to overcome this incurs additional loss in the approximation factor.

In order to overcome these practical issues, we propose an alternate MWU formulation which sacrifices the width bounds 
but escapes the binary search phase over $D$.
Eliminating the artificial width bound
makes it straightforward to extract a primal solution.
Moreover, our experiments on real world graphs suggest that width  is not a bottleneck for the running time of the MWU algorithm. Even more importantly, our alternate formulation 
naturally yields {\sc Greedy++} as we explain in the following.
 
\spara{Our MWU formulation.}   We can denote the LP $\textsc{Dual}(G)$ succinctly as follows:
\begin{align*}
\text{minimize }\qquad & D \\
\text{subject to}\qquad &\Bbf \fbf \leq D\one \\
&\fbf \in \PP
\end{align*}
where $\fbf$ is the vector representation of the all $f_e(v)$ variables,
$\Bbf \in \Rbb^{n \times 2m}$ is the matrix denoting the left hand side
of all constraints of the form $\displaystyle\sum_{e \ni v}f_e(v) \le D$.
$\one$ denotes the vector of $1$'s and $\PP$ is a polyhedral constraint
set defined as follows:
\begin{align*}
f_e(u) + f_e(v) & \geq 1 & \forall e=uv \in E \\
f_e(u) & \geq 0 & \forall e \in E,\ \forall v \in e.
\end{align*}
Note that for any $\fbf \in \PP$,
we have that the minimum $D$ satisfying $B\fbf \le D\one$
is equal to $\gnorm{B\fbf}{\infty}$.
This follows due to the non-negativity of $B\fbf$ for any $\fbf \in \PP$.
Now a simple observation shows that for any non-negative vector $\ybf$,
we can write
\[
\gnorm{\ybf}{\infty} = \max_{\xbf \in\Delta_n^+}\xbf^T\ybf
\]
where $\Delta_n^+ := \{\xbf\ge \zero: \one^T\xbf \le 1\}$.
Hence, we can now write $\textsc{Dual}(G)$ as:
\begin{align}
\min_{\fbf \in \PP} \gnorm{\Bbf \fbf}{\infty} &= \min_{\fbf \in \PP}\max_{\xbf \in\Delta_n^+} \xbf^T\Bbf \fbf \nonumber\\
 &= \max_{\xbf \in\Delta_n^+}\min_{\fbf \in \PP} \xbf^T\Bbf \fbf.\label{eq:dual_maxmin}
\end{align}
Here the last equality follows due to strong duality of the convex optimization.

The ``inner" minimization part of \eqref{eq:dual_maxmin} can be performed easily.
In particular, we need an oracle which, given a vector $\xbf$, solves
\[
C(\xbf) = \min_{\fbf \in \PP} \sum_{e=uv} x_u f_{e}(u) + x_v f_{e}(v).
\]
\begin{lemma}\label{lem:inner_min}
	Given a vector $\xbf$, $C(\xbf)$ can be computed in $O(m)$ time.
\end{lemma}
\begin{proof}
	For each edge $e=uv$, simply check which of $x_u$ and $x_v$ is smaller.
	WLOG, assume it is $x_u$.
	Then, set $f_{e}(u) = 1$ and $f_{e}(v) = 0$.
\end{proof}
We denote the optimal $\fbf$ for a given $\xbf$ as $\fbf(\xbf)$.
Now, using the above oracle, we can apply the MWU algorithm
to the ``outer" problem of \eqref{eq:dual_maxmin},
i.e., $\max_{\xbf \in \Delta_n^+}C(\xbf)$.
Additionally, to apply the MWU framework,
we need to estimate the width of this linear program.
The width for \eqref{eq:dual_maxmin} can be bounded by largest degree,
$d_{\max}$ of the graph $G$.
Indeed, we see in Lemma \ref{lem:inner_min} that $\fbf(\xbf)$ is a $0/1$ vector.
In that case, $\gnorm{B\fbf(x)}{\infty} \le d_{\max}$. 

We conclude our analysis of this alternative dual formulation of the DSP with the following theorem.

\begin{theorem}\label{thm:mwu}
	Our alternative dual formulation admits a MWU algorithm that outputs an $\fbf \in \PP$ such that $\gnorm{\Bbf \fbf}{\infty} \leq (1+\epsilon)\rho_G^*$.
\end{theorem}

For the sake of completeness, we detail the MWU algorithm and the proof of Theorem~\ref{thm:mwu} in Appendix~\ref{sec:appendix}.

Let us now view Charikar's peeling algorithm in the context of this dual problem.
In a sense, the greedy peeling algorithm resembles one ``inner" iteration  of the MWU algorithm,
where whenever a vertex is removed, its edges assign their load to it. Keeping this in mind, we designed {\sc Greedy++}  to add ``outer" iterations
to the peeling algorithm, thus improving the approximation factor arbitrarily with increase in iteration count.
By weighting vertices using their load from previous iterations, \textsc{Greedy++} implicitly performs
a form of load balancing on the graph, thus arriving at a better dual solution.

\section{Experiments}
\label{sec:exp} 
\subsection{Experimental setup} 
\label{subsec:setup} 

\newcommand\Tstrut{\rule{0pt}{2.6ex}}         
\newcommand\Bstrut{\rule[-0.9ex]{0pt}{0pt}}   

\begin{table}[!ht]
\begin{center}
\begin{tabular}{|l|c|c|} \hline
Name  & $n$  & $m$  \Tstrut\Bstrut\\ \hline 
web-trackers \cite{kunegis2013konect} &	40\,421\,974	& 140\,613\,762 \Tstrut\\
orkut \cite{kunegis2013konect} &	3\,072\,441 &	117\,184\,899 \\ 
livejournal-affiliations \cite{kunegis2013konect} & 10\,690\,276 &	112\,307\,385 \\ 
wiki-topcats & 1\,791\,489 &	 25\,447\,873 \\
cit-Patents	& 3\,774\,768 &	16\,518\,948 \\
actor-collaborations \cite{kunegis2013konect} &	382\,219 &	15\,038\,083\\
ego-gplus	 & 107\,614 &	12\,238\,285 \\ 
dblp-author	& 5\,425\,963	& 8\,649\,016 \\ 
web-BerkStan	& 685\,230 & 	6\,649\,470 \\ 
flickr \cite{zafarani2009arizona} & 80\,513	& 5\,899\,882 \\ 
wiki-Talk &	2\,394\,385	& 4\,659\,565 \\ 
web-Google &	875\,713 &	4\,322\,051 \\ 
com-youtube  &	1\,134\,890	& 2\,987\,624 \\ 
roadNet-CA	& 1\,965\,206	& 2\,766\,607 \\ 
web-Stanford & 	281\,903 &	1\,992\,636 \\ 
roadNet-TX &	1\,379\,917 &	1\,921\,660 \\
roadNet-PA	& 1\,088\,092 &	1\,54\,898 \\ 
Ego-twitter &	81\,306	 & 1\,342\,296 \\
com-dblp &	317\,080	& 1\,049\,866\\
com-Amazon &	334\,863 &	925\,872 \\ 
soc-slashdot0902	& 82\,168 &	504\,230 \\
soc-slashdot0811	& 77\,360 &	469\,180 \\ 
soc-Epinions &	75\,879 &	405\,740 \\ 
blogcatalog	\cite{zafarani2009arizona} & 10,312 &	333\,983 \\
email-Enron &	36\,692 &	183\,831\\
ego-facebook & 	4\,039 &	88\,234 \\ 
ppi \cite{stark2006biogrid} &	3\,890 &	37\,845  \Bstrut\\ \hline
twitter-retweet \cite{sotiropoulos2019twittermancer}  &	316\,662   &	1\,122\,070 \Tstrut\\ 
twitter-favorite \cite{sotiropoulos2019twittermancer}  & 226\,516	& 1\,210\,041 \\
twitter-mention \cite{sotiropoulos2019twittermancer} &	571\,157   &	1\,895\,094 \\ 
twitter-reply   \cite{sotiropoulos2019twittermancer}   & 	196\,697	& 296\,194 \Bstrut\\  
\hline 
soc-sign-slashdot081106  &	77\,350 &	468\,554 \Tstrut\\
soc-sign-slashdot090216 &	81\,867 &	497\,672 \\ 
soc-sign-slashdot090221	& 82\,140	& 500\,481 \\
soc-sign-epinions	            & 131\,828	& 711\,210 \Bstrut\\ \hline
\end{tabular}
\end{center}
\caption{\label{tab:datasets} Datasets used in our experiments. }
\end{table}

The experiments were performed on a single machine, with an Intel(R) Core(TM) i7-2600 CPU at 3.40GHz (4 cores), 8MB cache size, and 8GB of main memory.  We find densest subgraphs on the samples using binary search and maximum flow computations. The flow computations were done using C++ implementations
of the push-relabel algorithm~\cite{goldberg1988new},
HiPR\footnote{HiPR is available at \url{http://www.avglab.com/andrew/soft/hipr.tar}}.
We have implemented our algorithm {\sc Greedy++} and Charikar's greedy algorithm C++. Our implementations are efficient and our code is available publicly\footnote{Our code for \textsc{Greedy++} and the exact algorithm is available at the anonymous link \url{https://www.dropbox.com/s/jzouo9fjoytyqg3/code-greedy\%2B\%2B.zip?dl=0}}.

\begin{figure*}[!ht]
\centering
\begin{tabular}{@{}c@{}@{\ }c@{}} \includegraphics[width=0.49\textwidth]{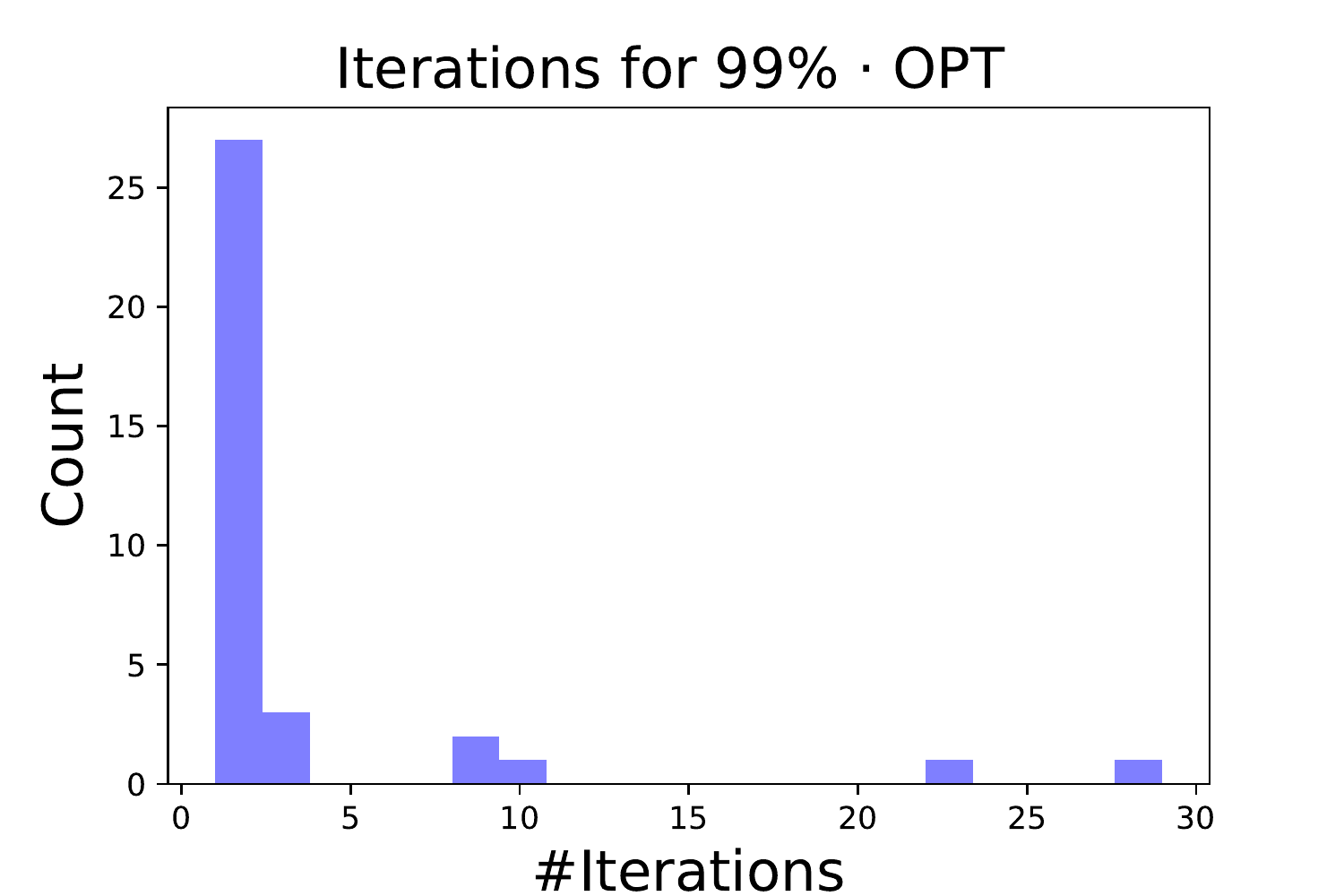} &
\includegraphics[width=0.49\textwidth]{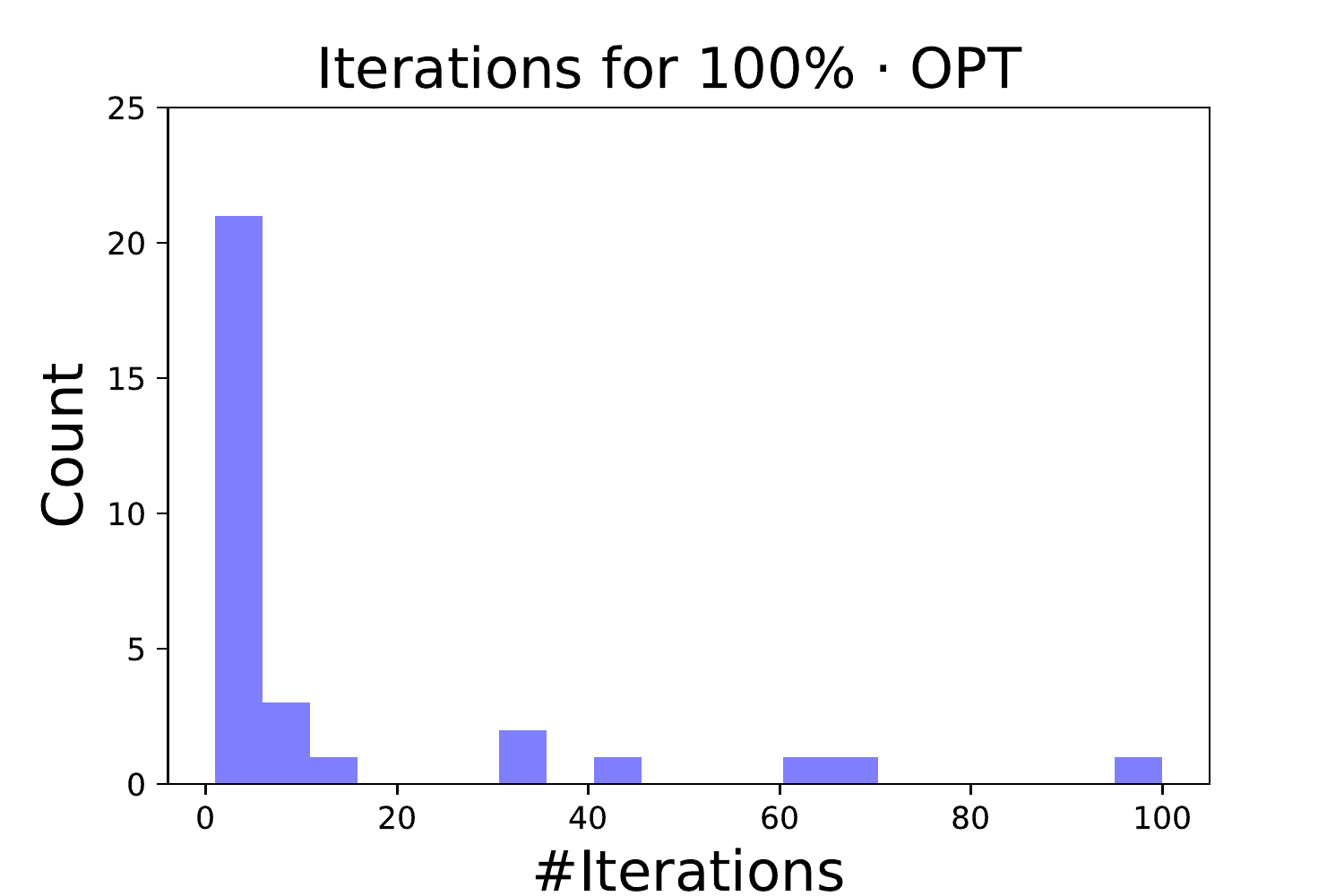} \hspace{5mm}  \\
(a) & (b) 
\end{tabular}
\caption{\label{fig:histograms} {\bf Number of iterations for {\sc Greedy++.}} Histograms of number of iterations to reach (a) 99\% of the optimum degree density, (b) the optimum degree density.}
\vspace{-4mm}
\end{figure*}

\begin{figure*}[!ht]
	\centering
	\begin{tabular}{@{}c@{}@{\ }c@{}} \includegraphics[width=0.49\textwidth]{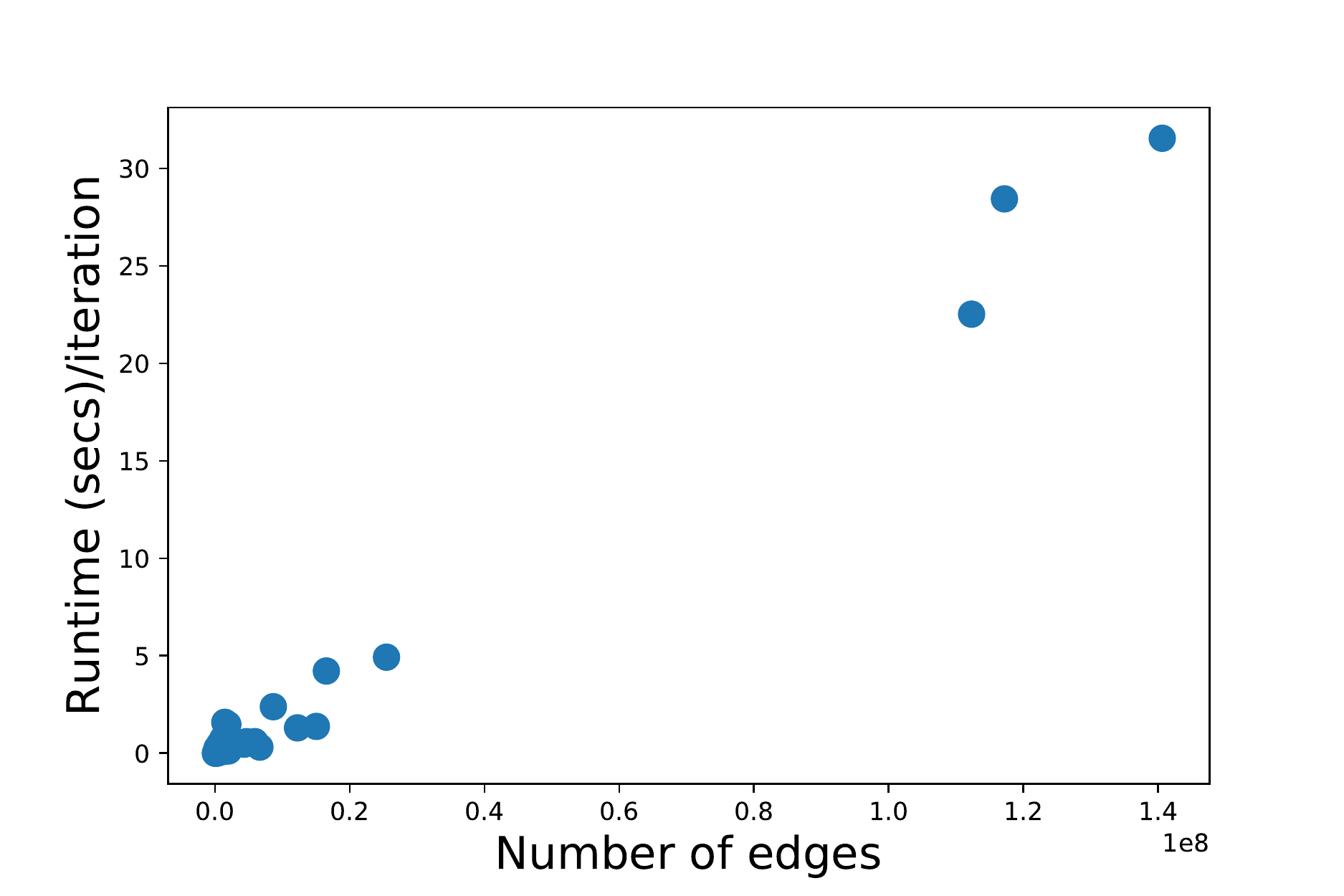} &
		\includegraphics[width=0.49\textwidth]{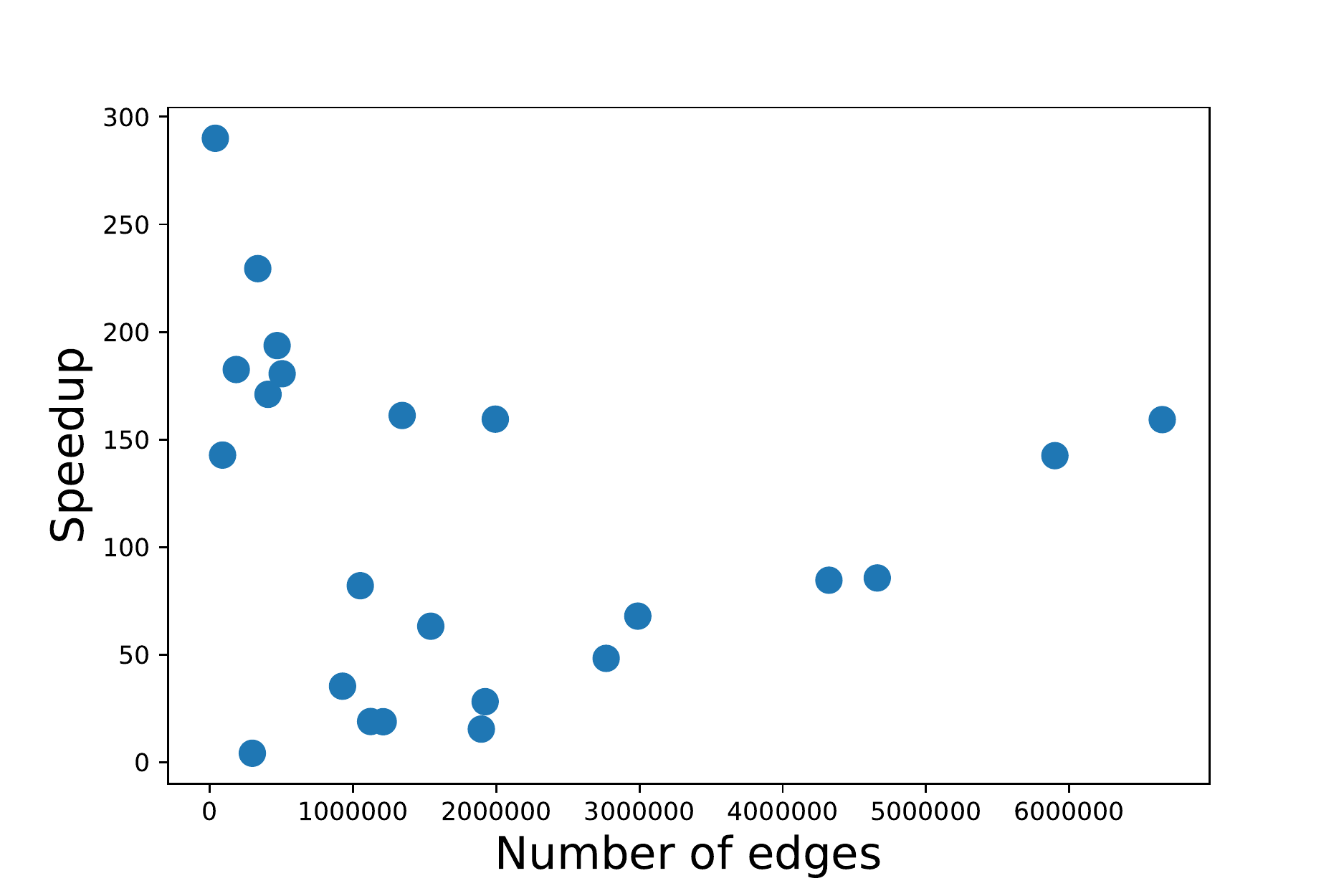} \hspace{5mm}  \\
		(a) & (b) 
	\end{tabular}
	\caption{\label{fig:scalability} {\bf Scalability.} (a) Running time in seconds of each iteration of {\sc Greedy++} versus the number of edges. (b) Speedup achieved by {\sc Greedy++} vs. number of edges in the graph. Specifically, the $y$-axis is the ratio of the run time of the exact max flow algorithm divided by the run time of {\sc Greedy++} that finds 90\% of the optimal solution. 
	}
	\vspace{-4mm}
\end{figure*}

We use a variety of datasets obtained from the Stanford's SNAP database \cite{leskovec2014snap}, ASU's Social Computing Data Repository \cite{zafarani2009arizona}, BioGRID \cite{stark2006biogrid} and from the Koblenz Network Collection \cite{kunegis2013konect}, that are shown in table Table~\ref{tab:datasets}. A majority of the datasets are from SNAP, and hence we mark only the rest with their sources. Multiple edges, self-loops are removed, and directionality is ignored for directed graphs. The first cluster of datasets are unweighted graphs. The  largest unweighted graph  is the \emph{web-trackers} graph with roughly 141M edges, while the smallest unweighted graph has roughly 25K edges. For weighted graphs, we use a set of Twitter graphs that were crawled during the first week of February 2018 \cite{sotiropoulos2019twittermancer}. Finally, we use a set of signed networks (\emph{slashdot}, \emph{epinions}). We remind the reader that while the DSP is NP-hard on signed graphs,  Charikar's algorithm does provide certain theoretical guarantees, see Theorem 2 in \cite{tsourakakis2019novel}.

\subsection{Experimental results} 
\label{subsec:findings} 

Before we delve in detail into our experimental findings, we summarize our key findings here: 
\begin{itemize}
\item Our algorithm {\sc Greedy++} when given enough number of iterations {\em always} finds the optimal value, and the densest subgraph. This agrees with our conjecture that running $T$ iterations of \textsc{Greedy++} gives a $1+O(\nicefrac{1}{\sqrt{T}})$ approximation to the DSP.  
\item Experimentally, Charikar's greedy algorithm  always achieves at least 80\% accuracy, and occasionally finds the optimum. 
\item For graphs on which the performance of Charikar's greedy algorithm is optimal, the first couple of iterations of {\sc Greedy++} suffice to deduce convergence safely, and thus act in practice as a certificate of optimality. This is the first method to the best of our knowledge that can be used to infer quickly the actual approximation of Charikar's algorithm on a given graph instance. 
\item When Charikar's algorithm does not yield an optimal solution, then {\sc Greedy++} within few iterations is able to increase the accuracy to 99\% of the optimum density, and by adding a few more iterations is able to find the optimal density and extract and optimal output. 
\item  When we are able to run the exact algorithm (for graphs with more than 8M edges, the maximum flow code crashes) on our machine, the average speedup that our algorithm provides to reach {\em the optimum} is 
144.6$\times$ on average, with a standard deviation equal to 57.4. The smallest speedup observed was 67.9$\times$, and the largest speedup 290$\times$.
Additionally, we remark that the exact algorithm is only able to find solutions up to an accuracy of $10^{-3}$ on most graphs.
\item The speedup typically increases as the size of the graph increases. In fact, the maximum flow exact algorithm cannot complete on the largest graphs we use.  
\item The maximum number of iterations needed to reach 90\% of the optimum is at most 3, i.e., by running two more passes compared to Charikar's algorithm, we are able to boost the accuracy by 10\%. 
\item The same remarks hold for both weighted and  unweighted graphs.
\end{itemize} 
 
%
%

 \begin{figure*}
         \centering
\begin{tabular}{@{}c@{}@{\ }c@{}}
\includegraphics[width=0.40\textwidth]{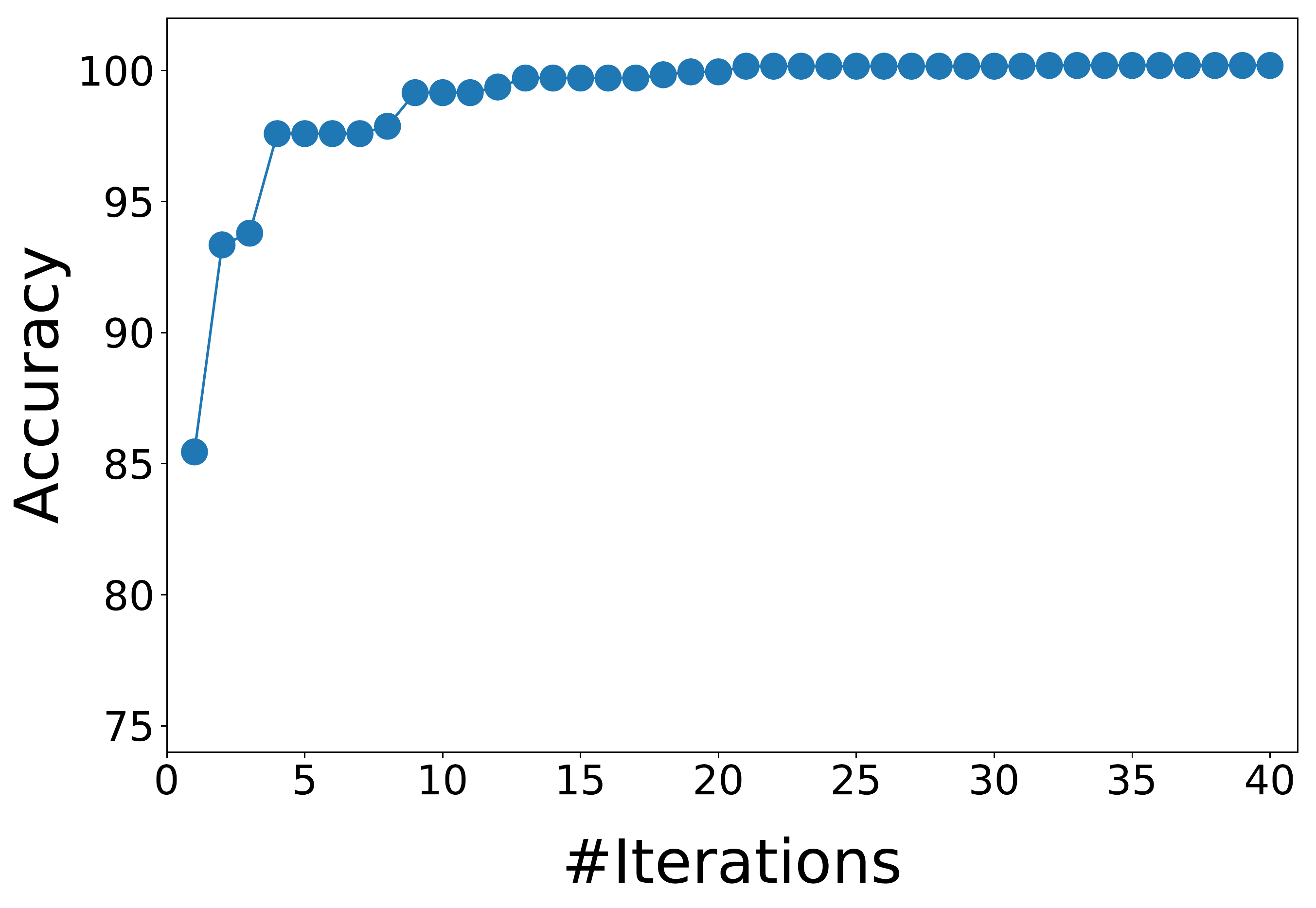}\hspace*{1cm} & \includegraphics[width=0.40\textwidth]{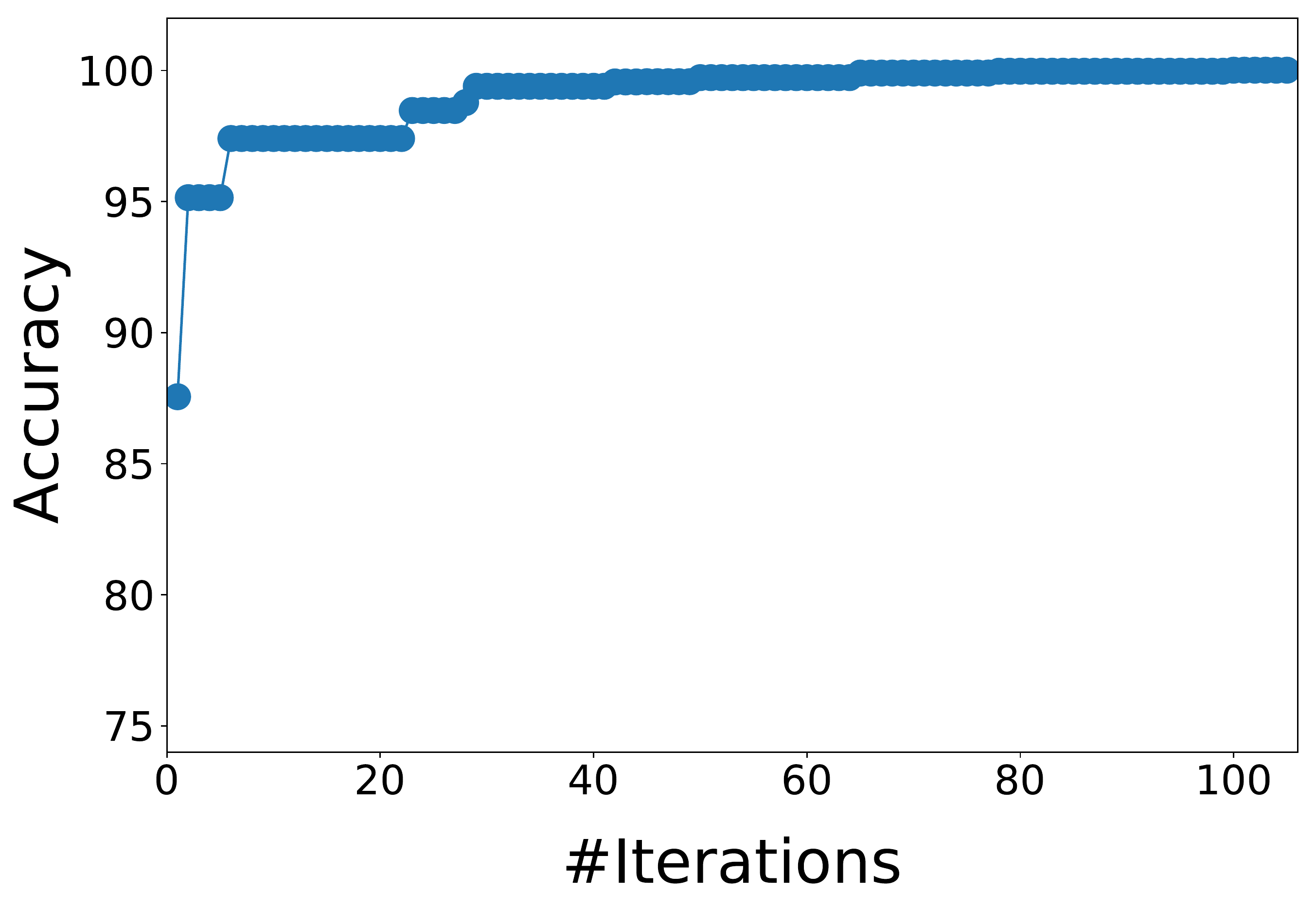} \\ 
(a) & \hspace*{1cm}(b)  \\ 
 \includegraphics[width=0.40\textwidth]{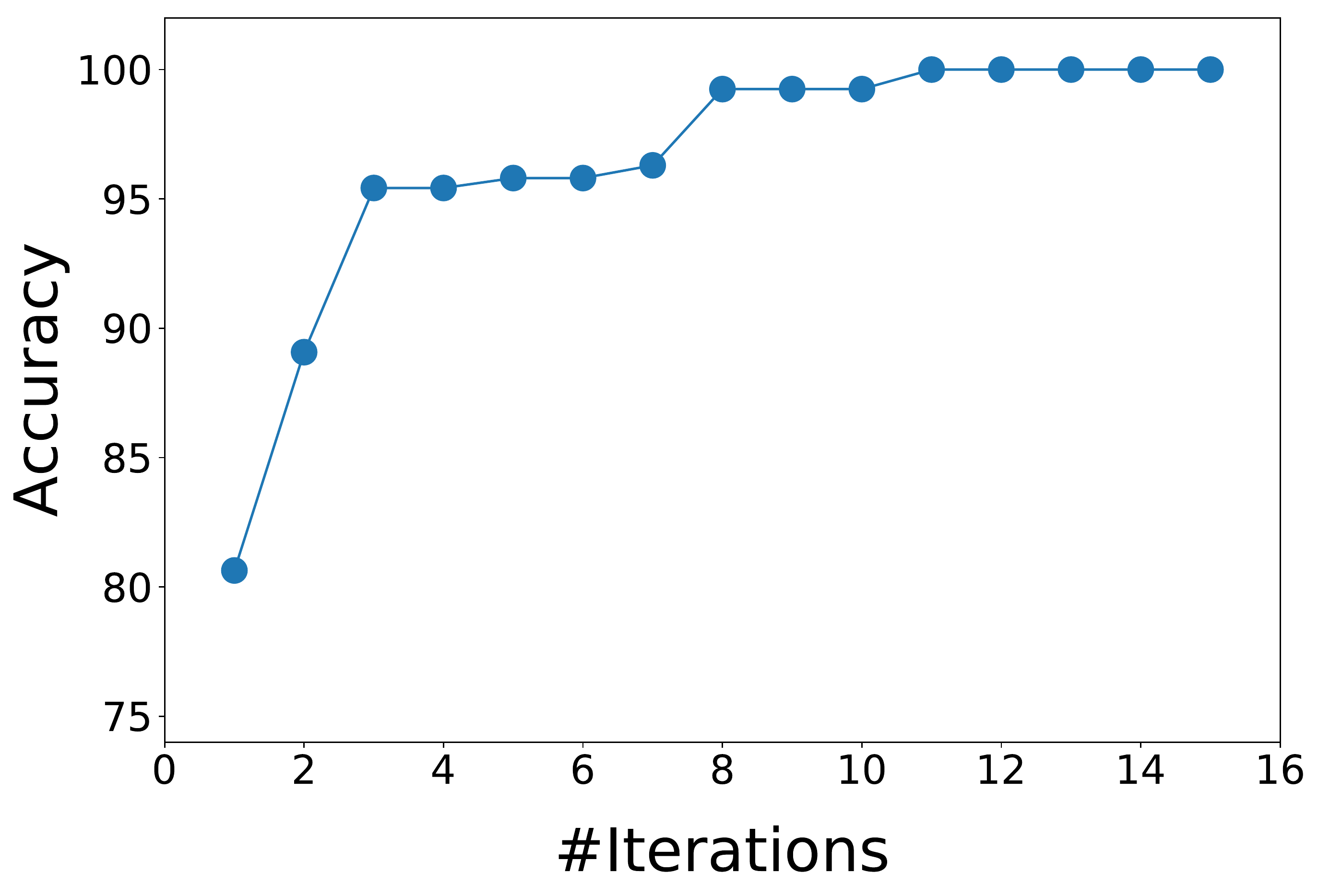}\hspace*{1cm} & \includegraphics[width=0.40\textwidth]{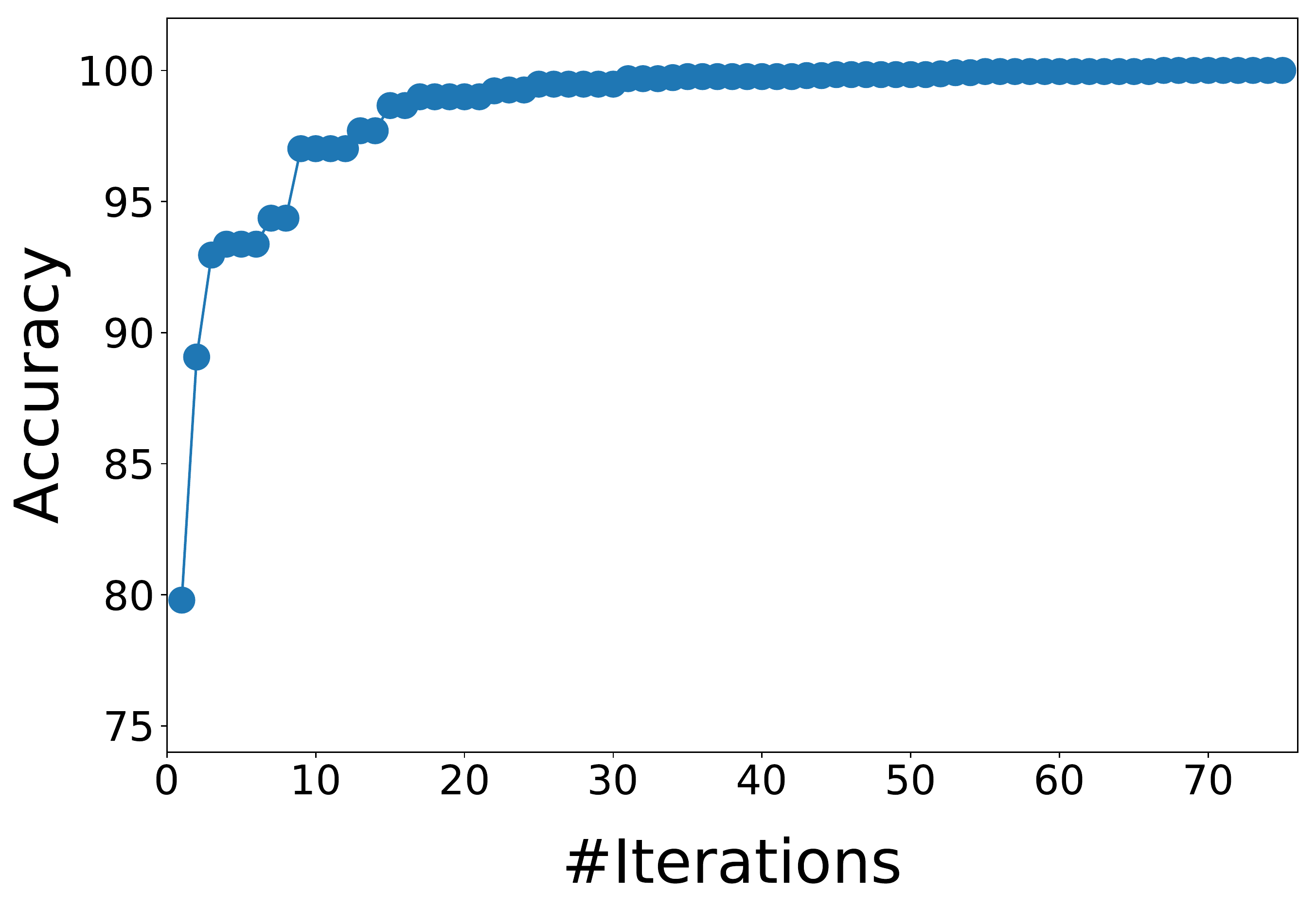} \\
(c) & \hspace*{1cm}(d) \\  
\includegraphics[width=0.40\textwidth]{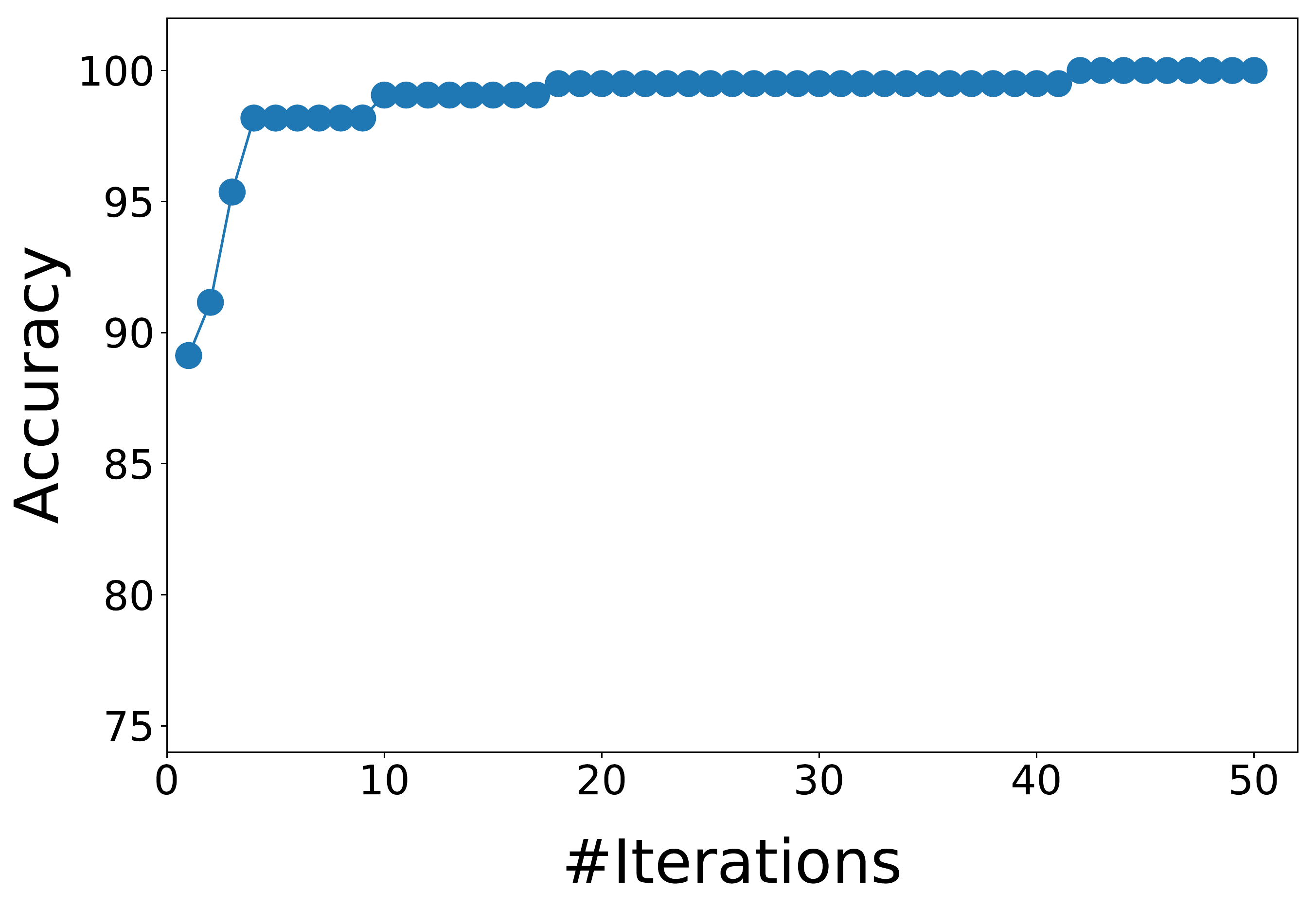}\hspace*{1cm} & \includegraphics[width=0.40\textwidth]{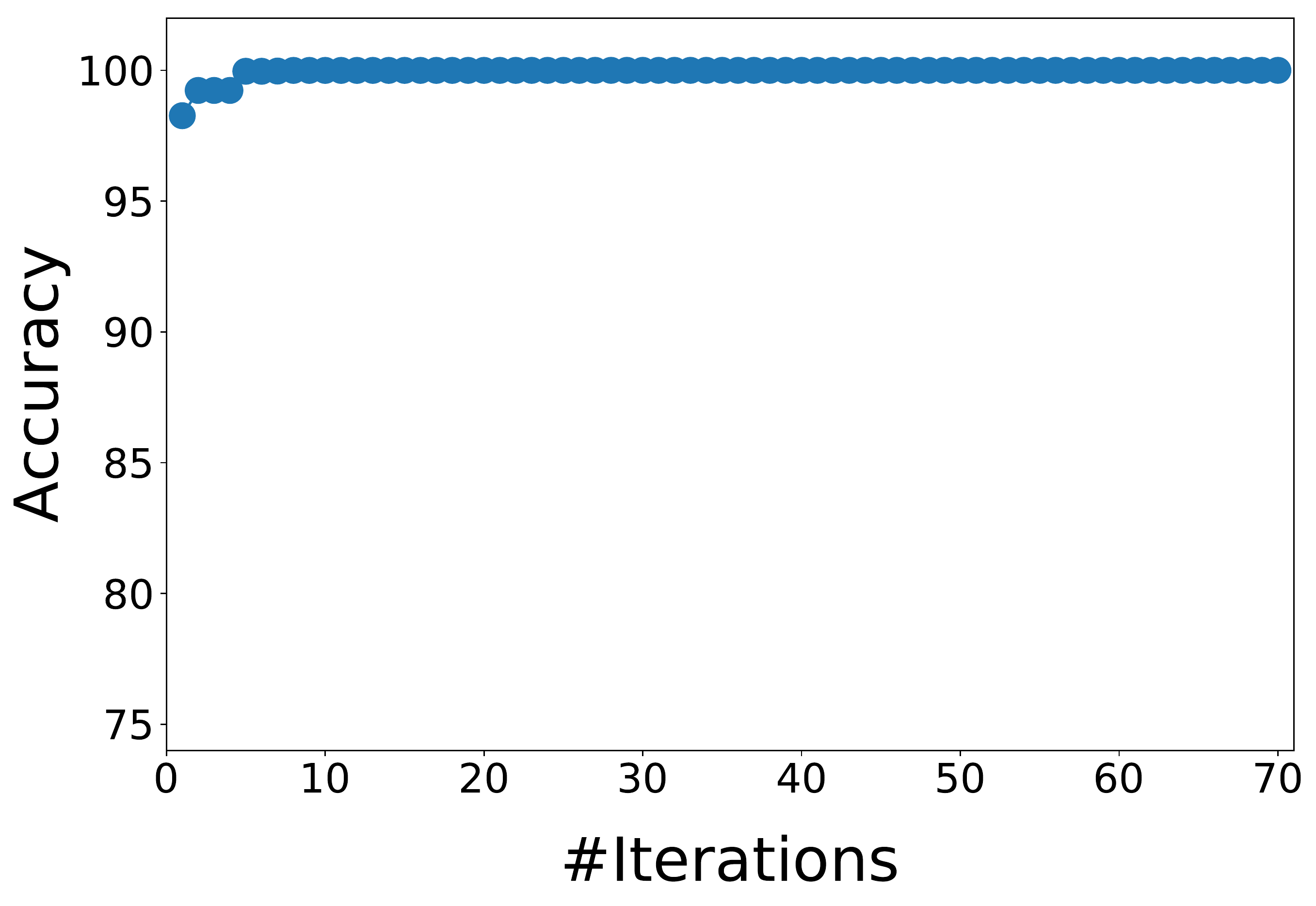} \\ 
(e) & \hspace*{1cm}(f) \\
 \includegraphics[width=0.40\textwidth]{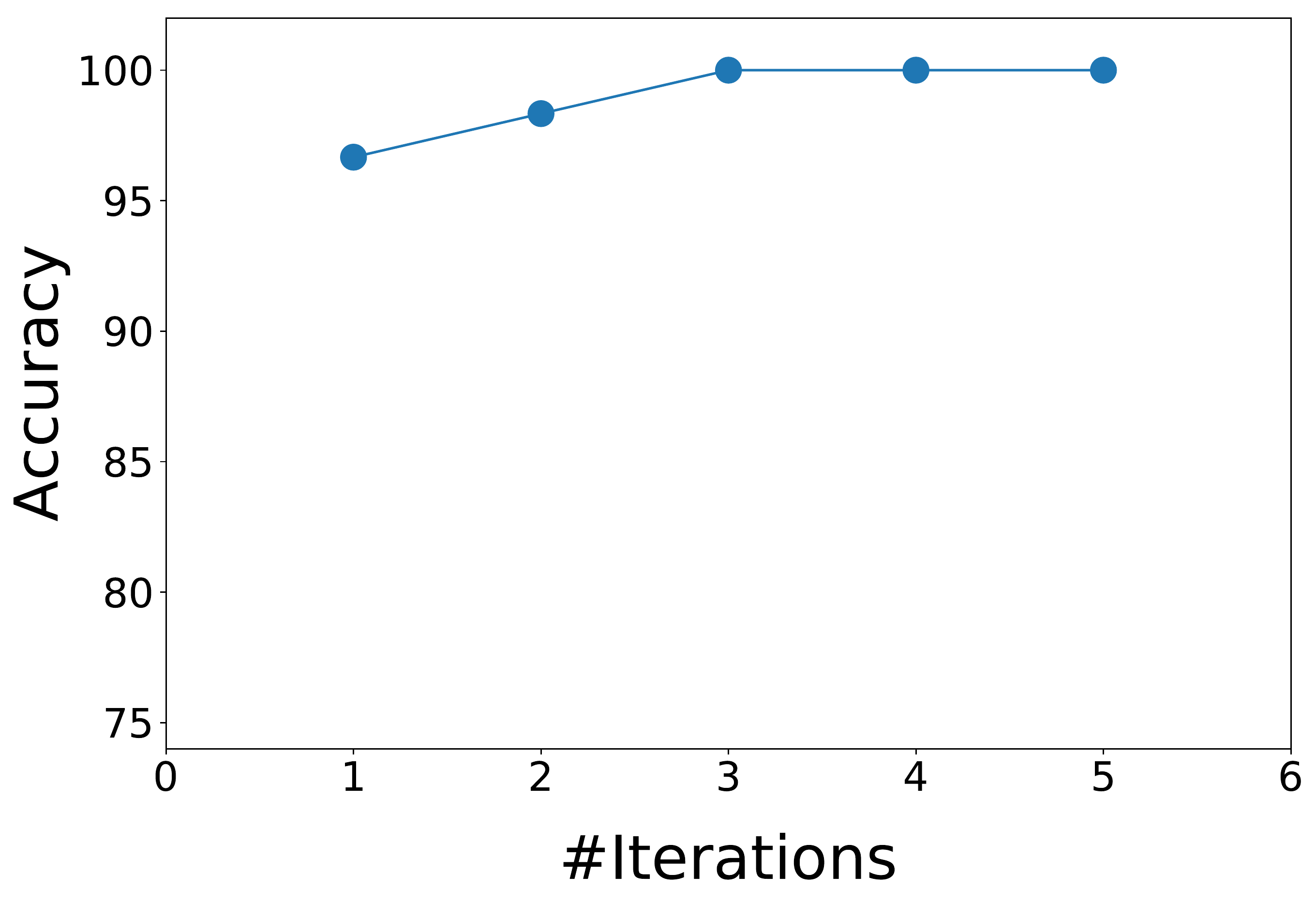}\hspace*{1cm} & \includegraphics[width=0.40\textwidth]{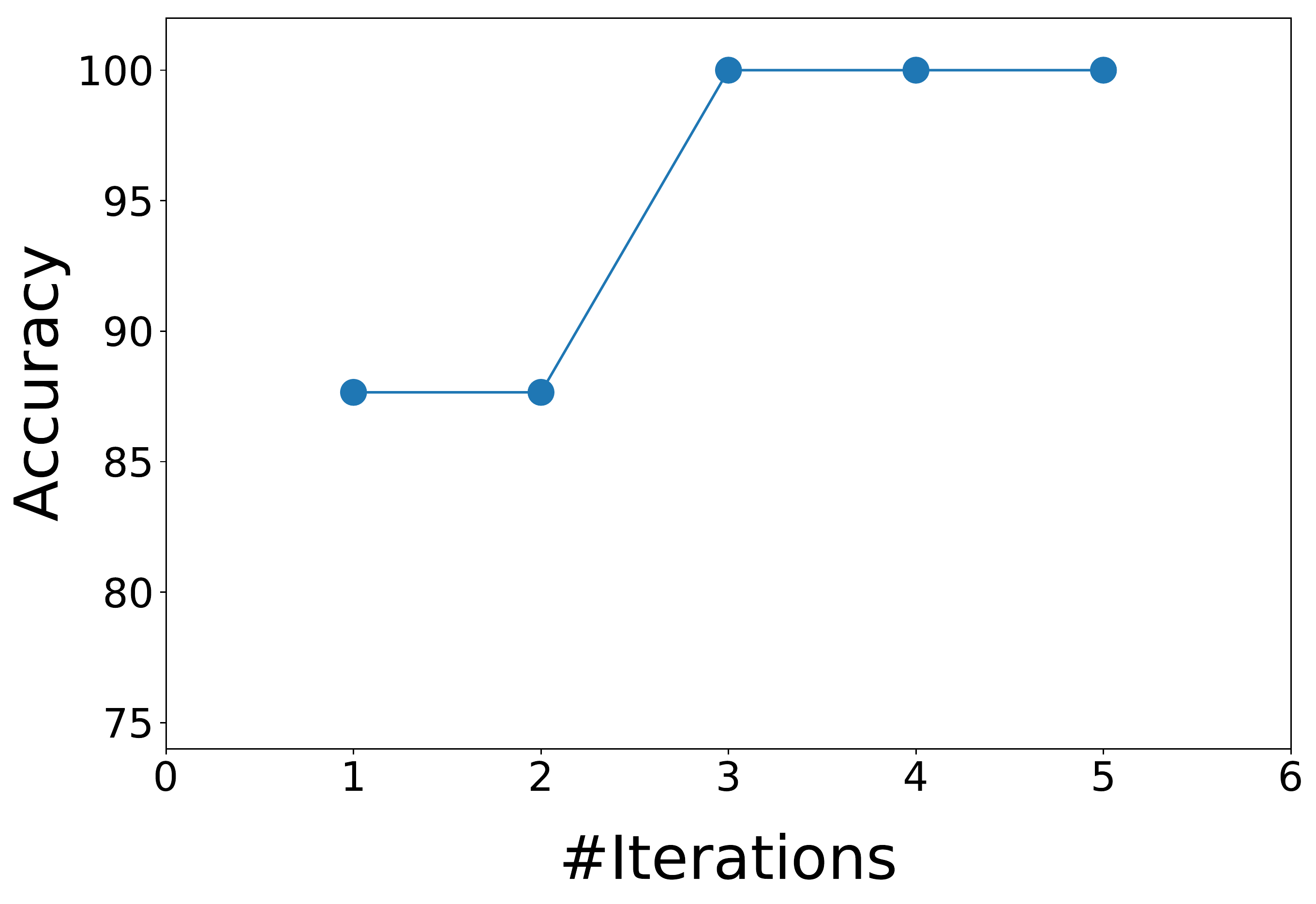} \\
(g) & \hspace*{1cm}(h)\\ 
\end{tabular}
\caption{\label{fig:convergence} Convergence to optimum as a function of the number of iterations of {\sc Greedy++}. (a) roadNet-CA, (b) roadNet-PA, (c) roadNet-TX, (d) com-Amazon, (e) dblp-author, (f) ego-twitter, (g) twitter-favorite, (h) twitter-reply.
Here, the \emph{accuracy} is given by $\bm{\frac{\rho(H_i)}{\rho_G^*}}$,
where $\bm{H_i}$ is the output of \textsc{Greedy++} after $\bm{i}$ iterations.}
\end{figure*}

\spara{Number of iterations.} We first measure how many iterations we need to reach 99\% of the optimum, or even the optimum. Figures~\ref{fig:histograms}(a),~(b) answer these questions respectively. We observe the impressive performance of Charikar's greedy algorithm; for the majority of the graph instances we observe that it finds a near-optimal densest subgraph. Nonetheless, even for those graph instances --as we have emphasized earlier-- our algorithm {\sc Greedy++} acts as a certificate of optimality. Namely, we observe that the objective remains the same after a couple of iterations if and only if the algorithm has reached the optimum. For the rest of the graphs where Charikar's greedy algorithm outputs an approximation  greater than 80\% but less than 99\%, we observe the following: for five datasets it takes at most 3 iterations, for one graph it takes nine iterations, and then there exist three graphs for which {\sc Greedy++} requires 10, 22, and 29 iterations respectively. If we insist on finding the optimum densest subgraph, we observe that the maximum number of iterations can go up to 100. On average, {\sc Greedy++} requires 12.69 iterations to reach the optimum densest subgraph.  

\spara{Scalability.}  Our experiments verify the intuitive facts that (i) each iteration of the greedy algorithm runs fast, and (ii) the exact algorithm that uses maximum flows is comparatively slow.  We constrain ourselves on the set of data for which we were able to run the exact algorithm. Figure~\ref{fig:scalability}(a) shows the time that each iteration of the {\sc Greedy++} takes on average (runtimes are well concentrated around the average) over the iterations performed to reach the optimal densest subgraph.    Figure~\ref{fig:scalability}(b) shows the speedup achieved by our algorithm when we condition on obtaining {\em at least}  90\% (notice that frequently the actual accuracy is greater than 95\%) of the optimal solution versus the exact max-flow based algorithm. Specifically, we plot  the ratio of the running times of the exact algorithm by the time of {\sc Greedy++} versus the number of edges. Notice that for small graphs, the speedups are very large, then they drop, and they exhibit an increasing trend as the graph size grows. For the largest graphs in our collection, the exact algorithm is infeasible to run on our machine.

\begin{figure}[!ht]
\centering
\includegraphics[width=0.45\textwidth]{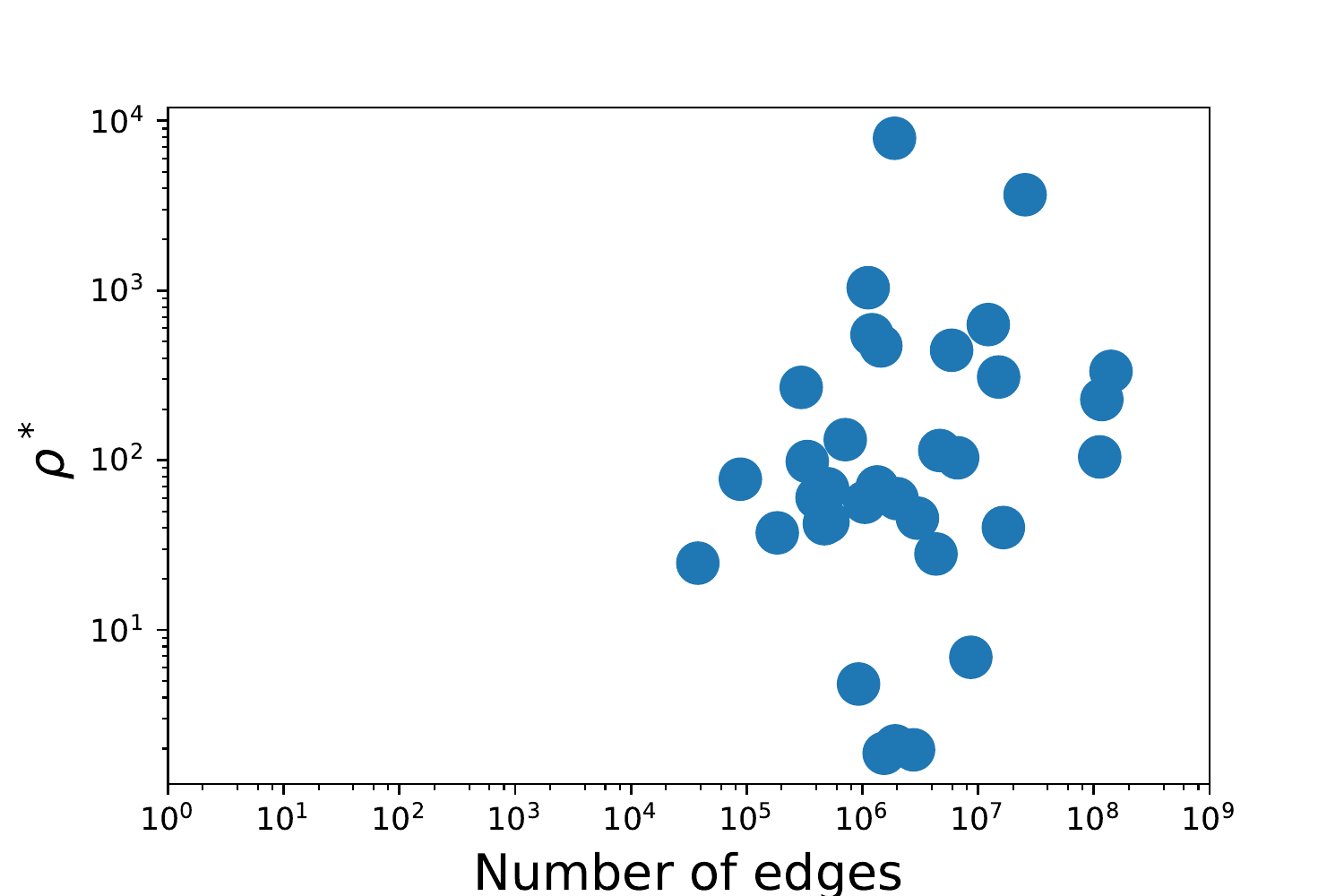}  
\caption{\label{fig:nopowerlaw} Log-log plot of optimal degree density $\rho^*$ versus the number of edges in the graph.  }
\end{figure}
 
\spara{Convergence.} Figure~\ref{fig:convergence} illustrates the convergence of {\sc Greedy++} for various datasets. Specifically, each figure plots the accuracy of {\sc Greedy++} after $T$ iterations versus $T$. The accuracy is measured as the ratio of the degree density achieved by {\sc Greedy++} by the optimal degree density. Figures~\ref{fig:convergence}(a),(b),(c),(d),(e),(f),(g),(h) correspond to the convergence behavior of roadNet-CA,   roadNet-PA, roadNet-TX,  com-Amazon,  dblp-author,   ego-twitter,   twitter-favorite, twitter-reply respectively. These plots illustrate various interesting properties of {\sc Greedy++} in practice. Observe Figure~\ref{fig:convergence}(e). Notice how {\sc Greedy++} keeps outputting the same subgraph for few consecutive iterations, but then suddenly around the 10th iteration it ``jumps'' and finds an even denser subgraph.  Recall that on average over our collection of datasets for which we can run the exact algorithm (i.e., datasets with less than 8M edges), {\sc Greedy++} requires roughly 12 iterations to reach the optimum densest subgraph.  For this reason we suggest running {\sc Greedy++} for that many iterations in practice. Furthermore, we typically observe an improvement over the first pass, with the exception of the weighted graph twitter-reply, where the ``jump'' happens at the end of the third iteration.

\spara{Anomaly detection.} It is worth outlining that {\sc Greedy++} provides a way to compute the densest subgraph in graphs where the maximum flow approach does not scale. For example, for graphs with more than 8 million edges, the exact method does not run on our machine. By running {\sc Greedy++} for enough iterations we can compute a near-optimal or the optimal solution. This allows us to compute a proxy of $\rho^*$ for the largest graphs, like orkut and trackers. We examined to what extent there exists a pattern between the size of the graph and the optimal density. In contrast to the power law relationship between the $k$-cores and the graph size claimed in \cite{shin2016corescope}, we do not observe a similar power law when we plot $\rho^*$ (the exact optimal value or the proxy value found by {\sc Greedy++} after 100 iterations for the largest graphs) versus the number of edges in the graph. This is shown in Figure~\ref{fig:nopowerlaw}. Part of the reason why we do not observe such a law are anomalies in graphs. For instance, we observe that small graphs may contain extremely dense subgraphs, thus resulting in significant outliers.

\section{Conclusion}
\label{sec:concl} 
In this paper we provide a powerful algorithm for the {\em densest subgraph problem}, a popular and important objective for discovering dense components in graphs. The main practical value of our {\sc Greedy++} algorithm is two-fold. First, by running few more iterations of Charikar's greedy algorithm we obtain (near-)optimal results that can be obtained using only maximum flows. Second, {\sc Greedy++} can be used to answer for first time the question ``Is the approximation of Charikar's algorithm on this graph instance closer to $\nicefrac{1}{2}$ or to $1$?'' without computing the optimal density using maximum flows. Empirically, we have verified that {\sc Greedy++} combines the best of ``two worlds'' on real data, i.e., the efficiency of the greedy peeling algorithm, and the accuracy of the exact maximum flow algorithm. 
We believe that {\sc Greedy++} is a valuable addition to the algorithmic toolbox for dense subgraph discovery that combines the best of two worlds, i.e., the accuracy of maximum flows, and the time and space efficiency of Charikar's greedy algorithm.  

We conclude our work with the following intriguing open question stated as a conjecture:

\begin{tcolorbox}
\begin{conjecture}
\label{conj} 
{\sc Greedy++} is a  $1+O(\nicefrac{1}{\sqrt{T}})$ approximation algorithm for the DSP, where $T$ is the number of iterations it performs. 
\end{conjecture}
\end{tcolorbox}

\bibliographystyle{abbrv}
\bibliography{acmart}	

\begin{appendix}

\section{Multiplicative Weights Update Algorithm} \label{sec:appendix}
	In this section, we give an algorithm to solve the zero-sum game $\max_{\xbf\in \Delta_n} \min_{\fbf \in \PP} \xbf^T\Bbf\fbf$, which corresponds to solving the dual of
	the densest subgraph problem, as described in Section~\ref{subsec:mwu}.
	Given that we have an oracle access to $\min_{\fbf \in \PP}\xbf^T\Bbf \fbf$, we can use the multiplicative weights update framework to get an $\eps$-approximation of the game \cite{freund1996game}.
	
	The pseudocode for the MWU algorithm is shown in Algorithm~\ref{alg:mwu}.
	
	\begin{algorithm}[!ht]
		\caption{Multiplicative Weight Update Algorithm}
		\label{alg:mwu}
		\begin{flushleft}
			\textbf{Input}: Matrix $\Bbf$, approximation factor $\eps$. 
			
			\textbf{Output}:  An approximate solution to the zero-sum game.
		\end{flushleft}
		\begin{algorithmic}[1]
			\State Initialize the weight vector as $w_i^{(1)} \gets 1$ for all $i \in [n]$
			\State Initialize $\eta \gets \nicefrac{\eps}{2\deg_{\max}}$
			\For{$t : 1 \rightarrow T$}
			\State $x_i^{(t)} \gets {w_i^{(t)}}/{\gnorm{\wbf^{(t)}}{1}}$ for all $i \in [n]$.
			\State Find $\fbf(\xbf^{(t)})$ using Oracle($\xbf^{(t)}$).
			\State Set $C(\xbf^{(t)}) \gets (\xbf^{(t)})^T\Bbf\fbf(\xbf^{(t)})$ 
			\State Let $\bbf_i^T\fbf(\xbf^{(t)})$ be the $i$-th element in $\Bbf\xbf^{(t)}$.
			\State Update the weights as \[w_i^{(t+1)} \gets w_i^{(t)}(1+\eta\bbf_i^T\fbf(\xbf^{(t)})).\]
			\EndFor
			\State Return $\dfrac{1}{T}\sum_{t\in [T]}C(\xbf^{(t)})$ as the solution.
			\end{algorithmic}
	\end{algorithm}

	To prove the convergence of Algorithm~\ref{alg:mwu}, we use the following theorem from \cite{arora2012multiplicative}. We modify it slightly to accommodate for the fact that
	the width of the DSP, $||Bf(x)||_{\infty}$, can be at most $\deg_{\max}$.
	In other words, the oracle can assign at most $\deg_{\max}$ edges to any particular vertex.
	
	\begin{lemma}[Theorem 3.1 from \cite{arora2012multiplicative}]
		Given an error parameter $\eps$, there is an algorithm which solves the zero-sum
		game up to an additive factor of $\eps$ using $O(W \log n / \eps^2)$ calls to \textsc{Oracle},
		with an additional processing time of $O(n)$ per call, where $W$ is the width of the problem.
	\end{lemma}
	
	Using the fact that our \textsc{Oracle} runs in $O(m)$ time (from Lemma~\ref{lem:inner_min}),
	and using $W = \deg_{\max}$, we get the following corollary.
	
	\begin{corollary}
	The Multiplicative Weight Update algorithm (Algorithm~\ref{alg:mwu}) outputs
	a $(1+\eps)$ approximate solution to the densest subgraph problem in time $O(m \deg_{\max} \log n / \eps^2)$.
	\end{corollary}
	
\end{appendix}

\end{document}